\documentclass[a4paper,12pt]{article}

\usepackage[margin=1in,footskip=0.25in]{geometry}
\usepackage{cite}
\usepackage{amsmath,amssymb,amsfonts}
\usepackage{amsthm}
\usepackage{thm-restate}
\usepackage{mathtools}
\usepackage{mathpartir}
\usepackage[inference]{semantic}
\usepackage{subcaption}

\usepackage{graphicx}
\usepackage{xcolor}
\usepackage{tikz}
\usepackage{tikz-cd}

\theoremstyle{plain}
\newtheorem{theorem}{Theorem}
\newtheorem{proposition}{Proposition}
\newtheorem{lemma}{Lemma}
\usepackage{cmll}
\usepackage{ebproof}

\theoremstyle{definition}
\newtheorem{definition}{Definition}
\newtheorem{remark}{Remark}
\newtheorem{example}{Example}

\makeatletter
\newif\if@draft\@drafttrue
\let\commentsize=\footnotesize
\if@draft%
\newcommand{\fo}[1]{{\commentsize\color{teal}[#1 -fo]}}
\newcommand{\udl}[1]{{\commentsize\color{blue}[#1 -udl]}}
\else
\newcommand{\fo}[1]{\ignorespaces}
\newcommand{\udl}[1]{\ignorespaces}
\fi
\makeatother

\usepackage{bbm}
\usepackage{mathrsfs}
\usepackage{amsmath}
\usepackage{amssymb}
\usepackage{stmaryrd}

\makeatletter
\newcommand{\xmultimap}[2][]{\ext@arrow 0359\multimapfill@{#1}{#2}}
\newcommand*{\multimapfill@}{%
	\arrowfill@\relbar\relbar\multimap}
\makeatother

\makeatletter
\newcommand*\sem[1]{\llbracket #1 \rrbracket}

\let\category=\mathcal

\newcommand*{\catM}{\category M}

\newcommand*{\@circmonadSymb}{{\mathcal T}_{\catM}}
\newcommand*{\circmonadNoSt@r}[3]{\@circmonadSymb(#1, #2, #3)}
\newcommand*{\circmonadSt@r}{\@circmonadSymb}
\newcommand*{\circmonad}{\@ifstar\circmonadSt@r\circmonadNoSt@r}

\newcommand*{\etax}{\overline{x}}

\newcommand*{\etaty}{\overline{\ty}}

\newcommand*{\eqdef}{\coloneqq}

\newcommand*{\Ob}[1]{\mathrm{ob}(#1)}

\newcommand{\rest}{\mathcal{RT}}

\newcommand{\Of}{\mathbb{O}_f}

\newcommand*\rTerms{\Lambda_r}

\newcommand*\ract[2]{#1\{#2\}}

\newcommand*\subst[3]{ #1 \subm{#3}{#2}}

\newcommand*{\subm}[2]{\{#1/#2\}}

\newcommand*{\abso}[1]{\lvert #1 \rvert}

\newcommand*{\fv}[1]{\mathsf{fv}(#1)}

\newcommand*{\length}[1]{\mathsf{len}(#1)}

\newcommand*{\morph}[1]{\text{hom}(#1)}

\newcommand*{\actr}[2]{{#1} \{#2\}}

\newcommand*{\qt}[1]{\mathsf{qt}\left(#1\right)}

\newcommand*{\intt}[1]{\mathsf{I}(#1)}

\newcommand*{\size}[1]{\mathtt{size}(#1)}

\newcommand*{\seq}[1]{\langle #1 \rangle}

\newcommand*{\seqdots}[3]{\seq{#1_{#2}, \dots, #1_{#3} } }

\newcommand*{\la}[1]{\lambda #1.}

\newcommand*{\toli}{  \Rrightarrow }

\newcommand*{\toexp}{ \Rightarrow }

\newcommand*{\ty}{a}

\newcommand*{\tyl}{\vec{a}}
\newcommand*{\tysx}{\overline{x}}
\newcommand*{\tysy}{\overline{y}}
\newcommand*{\tysa}{\overline{a}}
\newcommand*{\tysf}{\overline{f}}
\newcommand*{\tysg}{\overline{g}}
\newcommand*{\tysh}{\overline{h}}

\newcommand{\idm}{\textsf{id}}

\makeatletter
\newcommand*{\doublerightarrow}[2]{\mathrel{
  \settowidth{\@tempdima}{$\scriptstyle#1$}
  \settowidth{\@tempdimb}{$\scriptstyle#2$}
  \ifdim\@tempdimb>\@tempdima \@tempdima=\@tempdimb\fi
  \mathop{\vcenter{
    \offinterlineskip\ialign{\hbox to\dimexpr\@tempdima+1em{##}\cr
    \rightarrowfill\cr\noalign{\kern.5ex}
    \rightarrowfill\cr}}}\limits^{\!#1}_{\!#2}}}
\newcommand*{\triplerightarrow}[1]{\mathrel{
  \settowidth{\@tempdima}{$\scriptstyle#1$}
  \mathop{\vcenter{
    \offinterlineskip\ialign{\hbox to\dimexpr\@tempdima+1em{##}\cr
    \rightarrowfill\cr\noalign{\kern.5ex}
    \rightarrowfill\cr\noalign{\kern.5ex}
    \rightarrowfill\cr}}}\limits^{\!#1}}}
\newcommand{\colim@}[2]{\vtop{\m@th\ialign{##\cr
    \hfil$#1\operator@font lim$\hfil\cr
    \noalign{\nointerlineskip\kern1.5\ex@}#2\cr
    \noalign{\nointerlineskip\kern-\ex@}\cr}}}
\newcommand{\colim}{%
  \mathop{\mathpalette\colim@{\rightarrowfill@\textstyle}}\nmlimits@
}

\makeatother

\title{Linearization via Rewriting}
\author{Ugo dal Lago and Federico Olimpieri}

\begin{document}

\maketitle

\begin{abstract}
	We introduce the \emph{structural resource $\lambda$-calculus}, a new 
	formalism in which strongly normalizing terms of the $\lambda$-calculus can  
	naturally be represented, and at the same time any type derivation can be 
	internally rewritten to its linearization. The calculus is shown to be 
	normalizing and confluent. Noticeably, \emph{every} strongly normalizable 
	$\lambda$-term can be represented by a type derivation. This is the first example 
	of a system where the linearization process takes place \emph{internally}, while remaining 
	purely \emph{finitary} and \emph{rewrite-based}.
\end{abstract}

\section{Introduction}
Slightly less than a century since the pioneering contributions by Church 
and Gentzen, the study of proof theory, type systems and the $\lambda$-calculus
has produced a multitude of logical and computational formalisms in which  
proofs and programs can be represented and the process of either 
eliminating cuts or evaluating a program by forms of rewriting can be studied 
directly within the system at hand, often giving rise to confluence and 
normalization results (see, e.g.,~\cite{Girard73,Reynolds74,ML80,Parigot92}). 
The considered systems are often very expressive, which is why the aforementioned 
normalization properties become logically non-trivial, and can hardly be 
carried out combinatorially.

Since the mid-eighties, the situation sketched above has somehow changed: the 
advent of linear logic~\cite{Girard87} has allowed a finer structure to be attributed 
to the underlying computation process. By identifying \emph{structural} logical rules, 
and in particular \emph{contraction}, as the bottleneck in the normalization results, 
linear logic gave rise to the introduction of proof and type systems in which 
structural rules are \emph{very severely restricted} or \emph{not allowed at 
all}. As a consequence, normalization properties can be proved by purely 
combinatorial means: rewriting and cut-elimination have the effect of 
strictly reducing the size of the objects at hand. Systems of this kind, which we will 
call \emph{quantitative} --- simply to distinguish them from the somehow 
\emph{qualitative} systems we referred to in the previous 
paragraph --- include not only multiplicative linear logic~\cite{Girard87,DanosRegnier89}, but also 
non-idempotent intersection types~\cite{Gardner94,DeCarvalho18}, the resource 
$\lambda$-calculus~\cite{Boudol93,EhrhardRegnier08}, and certain 
type systems based on restrictions of exponentials such as the so-called
light logics~\cite{GSS92,Girard98,Lafont04}. In quantitative systems, there 
is an aspect of finiteness absent in qualitative systems, which makes 
them particularly suitable to the characterization of 
complexity classes and in general whenever a fine control over the use of 
resources is considered necessary. A proof or program, in quantitative systems, 
is an object that reveals \emph{everything} about its dynamics, \emph{i.e.}, about the evolution it 
will undergo as a result of the evaluation.

Since the early days of linear logic, there have been many attempts to relate 
quantitative and qualitative systems by studying how to transform proofs and 
programs in qualitative systems into equivalent objects in quantitative 
systems. Already in Girard's original work~\cite{Girard87} there is a trace of 
all this in the so-called \emph{approximation theorem}, which states that every 
\emph{cut-free} proof of multiplicative \emph{and exponential} linear logic can 
be approximated, in a very natural sense, by a proof of multiplicative linear 
logic. We should also mention the line of work about
\emph{linearizing} the $\lambda$-calculus~\cite{Kfoury00,AlvesFlorido03,
	AlvesFlorido05,AlvesVentura22}, in which one studies to 
what extent pure $\lambda$-terms can be linearized into linear terms, 
establishing on one hand that this process can be applied to all 
normalizing terms following their reduction sequence, and that on the other 
hand a relationship exists between the emerging process and typeability in non-idempotent 
intersection types.  We should also mention the crucial contribution of Ehrhard 
and Regnier~\cite{EhrhardRegnier08}, who through an inductive and very elegant definition, 
provide a notion of \emph{Taylor expansion} for any pure $\lambda$-term as the 
(infinite) set of its linear approximants. Continuing along the same vein, we 
must certainly mention the results on the extensional collapse of quantitative 
denotational models into qualitative models \cite{er:collapse}. The collapse induces a categorical understanding of the relationship between idempotent (\textit{i.e.}, qualitative) and non-idempotent (\textit{i.e.}, quantitative) intersection types. However, this result does not directly yield an \emph{explicit} program transformation.

In all these attempts, the infinitary aspects of qualitative systems 
are either neglected by considering only cut-free proofs, or reflected 
in the target of the transformation, like in the Taylor expansion. In the case 
of Kfoury's linearization~\cite{Kfoury00}, in particular, the process of computing the linear 
approximant of a pure $\lambda$-term is delegated to a notion of \emph{expansion} 
that, even if formally defined in terms of rewriting, is trivially not 
normalizing, because it consists in increasing the size of the right argument 
bag of an application, thus matching the number of arguments required by the 
related abstraction. In a certain sense, therefore, the linearization process 
remains fundamentally infinitary, similar to what happens in the Taylor 
expansion~\cite{EhrhardRegnier08}. The same holds for the remarkable recent 
contribution by Pautasso and Ronchi~\cite{PautassoRonchi23}, which 
introduces a quantitative version of the simple types in which derivations 
carry the same kind of quantitative information as in 
non-idempotent intersection types. The system turns out to be equivalent 
to the usual simply-typed $\lambda$-calculus in terms of typeability, but the
proof of this fact is carried out externally to the system, and is nontrivial.

Could linearization be made an \emph{integral part} of the underlying system? 
Is there a way to compute a linear version of a program (or proof) so that 
this transformation happens \emph{within} the system? The authors are strongly 
convinced of the interest of these questions and 
think that understanding to what extent all this is possible can help shedding
some light on the relationship between qualitative and quantitative systems.

The main contribution of the present paper is precisely to introduce a language and type 
system in which qualitative systems can be naturally embedded
\emph{and} an essentially quantitative fragment can be isolated. 
More specifically, terms typable in qualitative systems such as the
simply-typed $\lambda$-calculus can be compositionally embedded into the 
calculus. Linearization, this is the real novelty, becomes a
form of \emph{rewriting} turning any term into a 
linear approximation of it. The role of this rewriting process, called 
exponential reduction, is to eliminate the use of structural rules \emph{without} 
performing any $\beta$-reduction, the latter being part 
of the system itself. Crucially, exponential reduction is strongly normalizing.

The calculus thus obtained can be seen as an extension of 
Boudol's resource $\lambda$-calculus~\cite{Boudol93}, and we call it the \emph{structural 
resource $\lambda$-calculus}. Technically, the key ingredients in its 
definition are as follows. At the level of terms, there is both the possibility 
of having a nonlinear use of variables, but also of passing a bag of terms to 
an abstraction, in the style of the resource $\lambda$-calculus. Crucially, 
abstractions are labeled so as to identify the structural rules needed to fire
the corresponding redex. An intersection type assignment system is an integral 
part of the calculus and guarantees the termination of each sequence of 
$\beta$-reduction steps. As already mentioned, the latter is not the only 
available reduction rule, and is complemented in an essential way by exponential
reduction, giving rise to what we call \emph{structural reduction}.

In addition to the definition of the structural resource $\lambda$-calculus, the key 
technical contributions of this paper are proofs of confluence and strong 
normalization for structural reduction. Remarkably, every strongly normalizable
$\lambda$-term can be represented in a natural way by a structural resource
term, the embedding being based on \emph{idempotent} intersection types~\cite{dez:int}.

\section{A Bird's Eye View on Linearization and Intersection Types}\label{sec:birdseye}
The purpose of this section is to informally describe the process of 
linearization in the $\lambda$-calculus, anticipating some of the problems we 
have to face, at the same time introducing the non-specialist reader to the topic. 
We will proceed by analyzing a very 
simple example of a $\lambda$-term, which will also constitute a sort of 
running example in the rest of this manuscript.

Let us consider the $\lambda$-term $MN$, where
\[
M\eqdef \lambda x. w(x(xy))(xy)
\qquad
N\eqdef \lambda z.qzz
\]
To the above term we can attribute simple types in the obvious way:
\begin{equation}
w :o\rightarrow o\rightarrow o,y:o,q:o \rightarrow o\rightarrow o\vdash MN:o
\label{equ:exsimpletypes}
\end{equation}
The term $MN$ is not a linear term at all, however, in particular due to the 
nonlinear use of the variables $x$ and $z$ by the terms $M$ and $N$, 
respectively. Understanding the quantitative aspects of the dynamics of the 
term looking at its typing is not immediate, i.e., it is not possible in a direct way. 
By $\beta$-reducing the term, however, one realizes quite easily that the variable $y$, although 
occurring free only \emph{twice} in the term, will occur free \emph{six} times in its normal 
form. Depending on the use that $w$ and $q$ will make of their arguments, this 
number could grow or decrease.

Would it be possible to account for the calculations above \emph{from
within a type system}? The answer is positive: as is well-known, non-idempotent 
intersection types~\cite{Gardner94,DeCarvalho18} allow us to type 
the same term in the following way:
$$
   w :  \seq{\seq{o} \multimap \seq{o} \multimap o}, y : \seq{o}^6 , q : \seq{\seq{o} \multimap \seq{o} \multimap o}^4      \vdash MN : o
$$
where the sequence $ \seq{\ty_1,\dots, \ty_k}  $ stands for the (commutative but 
not idempotent) intersection type $ \ty_1 \cap \cdots \cap \ty_k $, and
$\seq{a}^n$ stands for $\seq{a,\ldots,a}$, where $ a$ appears exactly $ n$ times. Once a term is typed in non-idempotent intersection types, its type derivation tells us 
\emph{everything} there is to know about the dynamics of the term, including its length. 
If we want the size 
of the term itself to reflect the size of the derivation, that is, for the 
former to become a unique and canonical description of the latter, we must move 
to the so-called resource terms, where the second argument of each application 
is not a single term, but a \emph{bag} of terms. Resource terms correspond to \emph{linear approximations} of ordinary $\lambda $-terms. In our example, we obtain the resource terms $ s :=  \lambda x. w \seq{x \seq{x\seq{y,y}, x\seq{y,y}}} \seq{x\seq{y,y}} $ and $ t := \lambda z. q \seq{z} \seq{z}$, getting the following typing
for the application $s\seq{t}^4$:
\begin{equation}
          w :  \seq{\seq{o} \multimap \seq{o} \multimap o}, y : \seq{o}^6 , q : \seq{\seq{o}  \multimap \seq{o} \multimap o}^4     \vdash  s \seq{t}^4 : o
\label{equ:exresource}
\end{equation}

How can we compute well-behaved linear approximations (\textit{i.e.}, typed 
resource terms) of a given $\lambda$-term? If one is only interested
in the term's semantics, and given the aforementioned connection between 
$ \lambda$-term dynamics and its quantitative approximations, one possibility 
is to first compute the normal form of the term and then reconstruct the 
corresponding intersection type derivation. This can be seen as a rephrasing 
of Girard's \emph{approximation theorem}. However, this answer is not 
completely satisfactory, since it relies on the evaluation of the program. 
We can thus restate our question as: how do we compute linear approximations 
without reducing (to $ \beta$-normal form) the considered program?

Answering this question is the main goal of this paper. Our main contributions are 
precisely a language and type system in which \emph{all} the above-mentioned 
type derivations can live and where we can define an appropriate notion of 
rewriting which performs linearization. Below, we describe informally how the 
transition from (\ref{equ:exsimpletypes}) to (\ref{equ:exresource}) can 
occur \emph{within} the system. First, let us follow a naive proof-search heuristic to 
compute the non-idempotent intersection typing of our term. In order to do so, we first 
\emph{embed} $M $ and $N $ separately in the resource calculus as terms 
$  u = \la{x}  x \seq{ x \seq{y}  } \seq{x \seq{y}}  $ and  $ v = \la{z} q \seq{z} \seq{z}$ with  
\begin{align*}
w :  \seq{\seq{o} \multimap \seq{o} \multimap o}, y : \seq{o}^2 &\vdash u : \seq{\seq{o} \multimap o}^3 \multimap o  \\
 q :\seq{\seq{o} \multimap \seq{o} \multimap o} &\vdash v  : \seq{o}^2 \multimap o 
\end{align*}
Now, our proof search gets stuck, since the type of the argument
$v$ is not the one the function $u$ requires. This happens because 
$v$ records the presence of \emph{two} occurrences of $z$, while $ u $ 
only demands \emph{one} input. One possible solution would be to extend the 
unification algorithm for simple types to the intersection type setting, 
which is precisely what has been done in \cite{PautassoRonchi23}. However, we 
follow a different path. The types $\seq{o}^2 \multimap o   $ and $ \seq{o} \multimap o $ 
are different but \emph{coherent}, in the sense that their syntactic tree is the same 
up to the number of copies of types occurring inside intersections. In particular, the two are 
connected by a \emph{nested} structural rule, which is defined on the contraction 
$c_o:\seq{o} \to \seq{o,o}. $ We then extend the syntax and the typing of the 
resource calculus in order to account for (nested) structural rules. In particular, 
we shall allow the following type inference:
\[                 
\begin{prooftree}
\hypo{ q :\seq{\seq{o} \multimap \seq{o} \multimap o} \vdash v  : \seq{o}^2 \multimap o } \hypo{ c_o : \seq{o} \to \seq{o,o}}
\infer2{   q :\seq{\seq{o} \multimap \seq{o} \multimap o} \vdash p = \la{z^{c_o}} q \seq{z} \seq{z} : \seq{o } \multimap o}
\end{prooftree}            \]
from which we can clearly infer the typing $    w :  \seq{\seq{o} \multimap \seq{o} \multimap o}, y : \seq{o}^2,  q :\seq{\seq{o} \multimap \seq{o} \multimap o} \vdash u \seq{p}^3 : o .   $ The term $ u\seq{p}^3$ is not a \emph{linear} 
approximation of $ MN$, since we are allowing arbitrary structural rules to be used on bound variables. 
However, we shall recover the linear resource term $ s\seq{t}^4 $ as the \emph{normal form} of $ u\seq{p}^3 $ under 
a novel notion of labelled rewriting, following which we allow for the transformation
\[       
\la{z^{c_o}} q \seq{z} \seq{z} : \seq{o} \multimap o \to_{ c_o \multimap o  } \la{z} q \seq{z} \seq{z}  : \seq{o,o} \multimap o . 
\]
Note that we index the reduction with type morphisms connecting the typing of the two involved 
terms. We exploit an appropriate notion of \emph{transformation} on type derivations, that consists of an action of the morphism $ c_o \multimap o  $ on the term $ u , $ denoted as $   [c_o \multimap o] u $. This transformation will track all the occurrences of the bound variable $ x $ in $ u$, replacing their type $ \seq{o} \multimap o $ with $ \seq{o,o} \multimap o .$ In order to do so, the inputs of $x $ are \emph{duplicated} accordingly, obtaining the term $s $. The process consists then in a \emph{deterministic} and confluent notion of \emph{structural rule} elimination, which is performed only \emph{on-demand}. It is worth noting that the rewriting does not perform any substitution. In particular, if $ s \to_f t$ the two terms $s$ and $t$ are coherent and, hence, are (generally non-linear) approximations of the same ordinary $\lambda $-term.

\section{Some Mathematical Preliminaries}\label{sec:prelim}
Following~\cite{tao:gen}, we see intersection types as \emph{lists} of types. This will allow us to obtain 
a simple and elegant categorical framework, where intersection plays the role of a 
(strict) tensor product, and morphisms are given by structural rules (\textit{i.e.}, permutations, 
copying, and erasing). In order to properly establish the formal framework, we first need 
to introduce categories of integers, from which the list construction naturally arises.

 \subsection{Categories from Natural Numbers}

We consider the category $ \Of^c$ whose objects are natural numbers and in which
$\Of^c(m, n):=[n]^{[m]}$, where $ [n] = \{ 1, \dots, n \} $ for $n \in \mathbb{N}$.
As a consequence,  $[0]=\emptyset$ (so the number $0$ is the initial object, 
$\emptyset$ being the unique map in $[m]^{[0]}$). We generally denote composition of morphisms as either $ g \circ f $ or $ gf$. The category $ \Of^c $ is symmetric \emph{strict} monoidal (cocartesian in particular), 
with tensor product given by addition. We define the subsets 
of \emph{symmetries} of $ \Of^c (m,n) $ as $ \Of^{l}(m, n) = 
\{  \alpha \in \Of (m, n) \mid \alpha \text{ is bijective}     \}$.
Evidently, $ \Of^l (m,n)$ determines a subcategory of $ \Of^c$, the category 
$ \Of^l$ of \emph{symmetries} (or \emph{permutations}). In what follows, 
we work with the parametric category 
$ \Of^\spadesuit$, where $ \spadesuit \in \{ c,l \} . $


\subsection{Lists}

We often work with (ordered) lists whose elements are taken from a given set $X$.
Such lists are ranged over by metavariables like $ \vec{a}, \vec{b}, \vec{c} \dots$, and
one such list can also be written explicitly as $ \seq{a_1, \dots, a_k} $, where
$k \in \mathbb{N}$ and $ \ty_1,\dots, \ty_k \in X$. The \emph{list concatenation} of $\vec{a}$ and $\vec{b}$ is 
indicated as $ \tyl \oplus \vec{b} . $ We write $\abso{\tyl}  $ to denote the length of $ \tyl$. 

Given a list $ \vec{a} = \seqdots{\ty}{1}{k} $ and a function $ \alpha : [h] \to [k] $,
the expression $ \vec{a}^{[\alpha]}$ stands for the list of length $h$ defined as
$\seqdots{\ty}{\alpha(1)}{\alpha(h)}$. This turns $\alpha$ into a contravariant semigroup
action on $\vec{a}$:  given $ \length{\vec{a}} = n, \alpha : [m] \to [n]  $ and $ \beta : [l] \to [m]  $ we have that
$(\vec{a}^{[\alpha]})^{[\beta]} = \vec{a}^{[\alpha \beta]}$.


Given  a small category $ A $, we define its $ \spadesuit $-\emph{monoidal completion}, 
$\oc^{\spadesuit}(A)  $ as follows:
\begin{itemize}
\item $ \Ob{\oc^{\spadesuit}(A)} = \{  \seqdots{\ty}{1}{k} \mid k \in \mathbb{N} \text{ and } \ty_i \in \Ob{A}       \} .$
\item $ \oc^{\spadesuit}(A)(\seqdots{\ty}{1}{k}, \seqdots{b}{1}{l}) =$
 \[\{ \seq{\alpha; f_1, \dots, f_l} \mid \alpha \in \Of^{\spadesuit}(l, k) \text{ and } f_{i}\in A(\ty_{\alpha(i)},b_{i}) \text{ whenever } i \in [l] \}. \]
\end{itemize}
The category $ \oc^l(A) $ is \emph{symmetric} strict monoidal, while $  \oc^c(A) $ 
is \emph{cartesian} strict monoidal. In both cases, the tensor product is given by list concatenation.

\emph{Ground morphisms} are morphisms of the shape $\seq{\alpha; \idm, \dots, \idm} $, thus representing the composition of permutations, weakenings and contractions without nesting. Given a list $ \tyl $, an integer $  k \in \mathbb{N}$ and  a morphism $ \alpha \in \Of^{\spadesuit} ([\length{\tyl}] , [k]) , $ we define the \emph{ground morphism} induced by $\alpha $ as the morphism $  \seq{\alpha; \vec{\idm}} \in  \oc^{\spadesuit}(A)(\tyl, \tyl^{[\alpha]})  $. We shall often abuse the notation and denote $ \seq{\alpha; \vec{\idm}}$ as just $\alpha $. The family $ \alpha_{\tyl}  $ for $\tyl \in  \oc^{\spadesuit}(A) $ is, moreover, a \emph{natural} in $ \tyl$. We use $ \sigma, \tau \dots$ for ground morphisms induced by permutations.

\begin{example} We give some examples of morphisms in $ \oc^c(A)  $, where $   A = \{ a,b,c \}  . $ The following is a diagonal morphism (that we also call a \emph{contraction}) over $A$:
\[     c_\ty =   \seq{   \alpha;  \idm_\ty, \idm_\ty          } : \seq{\ty} \to \seq{  \ty,\ty   }                      \]  where $\alpha$ is the only map from $[2] $ to $[1] $. For all $ n\geq 1 $ we have a $ n$-ary contraction \[ c_{\ty}^{n} : \seq{\ty} \to \overbrace{\seq{\ty, \dots, \ty}}^{n \text{ times}} \] where for $ n = 1 $ we have that $ c^1 = \idm_\ty $. Projections (that we also call \emph{weakenings}) are given as follows:
\[      \pi_{\ty_i}  = \seq{\alpha;  \idm_{\ty_i}      }  : \seq{\ty_1,\ty_2} \to \seq{\ty_i}                        \]
where $ \alpha : [1] \to [2] $ maps $1 $ to $1 $. We denote as $\mathtt{T}_{\tyl} : \tyl \to \seq{} $ the terminal morphism.

\end{example}

Given sequences of lists $ \gamma = \tyl_1, \dots, \tyl_n $ and $ \delta = \vec{b}_1,\dots, \vec{b}_n $ we define their \emph{tensor product} as pointwise list concatenation: $ \gamma \otimes \delta = \tyl_1 \oplus \vec{b}_1, \dots, \tyl_n \oplus \vec{b}_n . $ Typing contexts in our system will consist of sequences of lists. The tensor product between them will model context concatenation. The ground morphism induced by $ \gamma = \tyl_1, \dots, \tyl_n $ consists of the sequence of ground morphisms of $\tyl_i $. Abusing the notation, we will use greek letters to denote also ground morphisms induced by sequences. 

\begin{remark}\label{rem:struct}
Consider the morphism on lists $  \seq{ \alpha ; \vec{f}}  = \seq{  \alpha ; f_1,\dots, f_l   } :  \tyl = \seqdots{\ty}{1}{k} \to \vec{b} = \seqdots{b}{1}{l} .  $ By definition, we have that $   f_i : \ty_{\alpha (i)} \to b_i  $ with $ i \in [l] . $ We remark that we can always factorize $ \seq{  \alpha ; f_1,\dots, f_l   }  $ as follows:
\[   \seq{\alpha; \vec{f}}    =   \vec{f} \circ \alpha_{\tyl} \]
with $  \vec{f}: \tyl^{[\alpha]} \to \vec{b}  $ and $  \alpha_{\tyl} : \tyl \to \tyl^{[\alpha]}$, defined as $ \alpha_{\tyl} = \seq{\alpha; \idm_{\ty_{\alpha(1)}}, \dots, \idm_{\ty_{\alpha(l)}}} $ and $ \vec{f} =   \seq{\idm; f_{1}, \dots, f_{l} }  .$
We call $     \alpha_{\tyl}       $ the \emph{ground fragment} of $ \seq{\alpha; \vec{f}} $, denoted as $ \mathsf{ground}(\seq{\alpha; \vec{f}}) $ and $ \vec{f} $ the \emph{nested fragment} of it, denoted as $ \mathsf{nest}(\seq{\alpha; \vec{f}})$. We can extend the factorization to sequences of morphisms $  \theta = \seq{\alpha_1; \vec{f}_1}, \dots, \seq{\alpha_n; \vec{f}_n}   $ in the natural way.
\end{remark}

\subsection{Labelled Transition Systems} 
In the rewriting systems we deal with in this paper, reduction steps are \emph{labeled}.  A \emph{labelled reduction system} consists of a labelled transition system 
$ \mathsf{R} = (\Lambda, \mathbb{S}, \to) $ whose labels form a partial monoid 
$ \mathbb{S} =  (S, \cdot, 1)  . $\footnote{We consider partial monoids because the labels of our rewriting systems will be morphisms in appropriate categories of types, where the monoidal product we want to consider is the composition of those morphisms. } Given a labelled reduction system $\mathsf{R} = (\Lambda, (S, \cdot, 1) , \to)$, we abuse the notation by also writing $ s \to s' $ to denote the relation $ \to \subseteq \Lambda \times \Lambda $ s.t.  $ s \to s'$ when there exists $  l \in S$ with $ s\to_l s'  $. Given a labelled reduction system $\mathsf{R} = (\Lambda, (S, \cdot, 1) , \to)  $, we say that $\mathsf{R} $ is \emph{normalizing} when, given $ s \in \Lambda ,$ there exists a reduction 
sequence $s = s_0 \to s_1  \to \dots \to s_n   $ such that  $s_n \nrightarrow  $, 
and in this case we call $s_n $ a \emph{normal form} of $s$. We say that $\mathsf{R} $ is 
\emph{strongly normalizing} (or \emph{terminating}) when it
is impossible to form an infinite reduction sequence  
$s_0 \to s_1  \to s_2 \to \ldots$

\begin{definition}[Reflexive and Transitive Extension]
Let $ \mathsf{R} = (\Lambda, (S, \cdot, 1), \to) $ be a labelled reduction system. 
We can thus  build its \emph{transitive and reflexive extension} 
$\mathsf{R}^\ast = ( \Lambda, (S, \cdot,1) , \to^\ast ) $ where 
$ \to^\ast $ is defined inductively by the following rules:
\[       \begin{prooftree}
\hypo{  s \to_h t                   }\hypo{ t \to^\ast_l v  } \hypo{ l \cdot h \downarrow }
\infer3{   s \to^\ast_{l \cdot h} v        } 
\end{prooftree}
       \qquad 
       \begin{prooftree}
       \infer0{    s \to_{1}^\ast   s            }
       \end{prooftree}                         \]
\end{definition}

\newcommand{\from}{\leftarrow}
\begin{definition}[Labelled Confluence]
Let $  \mathsf{R} = (\Lambda, (S, \cdot, 1) , \to)  $ be a labelled reduction system. 
We say that $ \mathsf{R} $ enjoys \emph{confluence} if whenever $t_1\prescript{\ast}{l_1}{\from}\;s \to_{l_2}^\ast t_2 $ 
there exist $ h_1,h_2 \in S $ and $ u\in \Lambda $ 
such that $  t_1\to_{h_1}^\ast u\;\prescript{\ast}{h_2}{\from} t_2$ with $ h_1 \cdot f_1 = h_2 \cdot f_2 $.
If the hypothesis becomes  $t_1\prescript{}{l_1}{\from}\;s \to_{l_2} t_2 $, we say
$\mathsf{R}$ is \emph{locally confluent}.
\end{definition}
The well-known Newman Lemma continues to hold even in presence of labels: 
local confluence together with strong normalization imply confluence, something which can 
be proved exactly as for abstract reduction systems~\cite{Terese03}.

\section{Types and Morphisms}
Intersection types are often defined \emph{equationally}: the intersection type constructor $ \ty \cap b $ is given as a monoidal product satisfying some equations: associativity, commutativity and possibly idempotency. In this paper, we take a more refined viewpoint: the intersection connective is still a strictly associative tensor product, but we replace equations by \emph{concrete} operations on types. For instance, commutativity is replaced by appropriate \emph{symmetries} $ \sigma_{a,b} : a \cap b \to b \cap a . $ Idempotency becomes a \emph{copying operation} $ c_{\ty} : \ty \to \ty \cap \ty . $ In order to formalize this idea, we identify the  intersection type $ \ty_1 \cap \cdots \cap \ty_k $ with the list $ \seqdots{\ty}{1}{k} . $ Operations on intersection types become then structural operations on lists, seen as a categorical tensor product (see Section \ref{sec:prelim}). Moreover, intersection types will appear on the left side of an application $ \seqdots{\ty}{1}{k} \multimap \ty $ only, following the well-established tradition of \emph{strict intersection types} \cite{bakel:strict}. This allows us to define a syntax-directed type system.

We define the set of \emph{resource types} ($\rest$) and the set of \emph{intersection types} ($ \oc(\rest)$) by the following inductive grammar:
\begin{align*}
\rest &\ni \ty, b ::= o \mid \tyl \multimap \ty \\
\oc(\rest) &\ni \tyl ::=  \seqdots{\ty}{1}{k} \qquad (k \in \mathbb{N})                       
\end{align*}
We define a  notion of \emph{morphism} of resource types by induction in Figure \ref{fig:rtctx}, parametrically on $\spadesuit \in \{ l,c \} $. \emph{Morphisms} of resource types can be seen as proof-relevant \emph{subtyping}. Composition and identities of these morphisms are defined by induction in a natural way, exploiting the definition of composition induced by the $\oc^{\spadesuit} (-) $ construction. Composition is associative and unitary. We can then define a category $ \rest^{\spadesuit} . $ $\rest^l $ (resp. $ \rest^c$) is said to be the category of \emph{linear} (resp. \emph{cartesian}) resource types in which objects are resource types and morphisms are \emph{linear} (resp. \emph{cartesian}) morphisms of resource types. We remark that the inclusion of categories $\Of^{l} \hookrightarrow \Of^{c}  $ induces an inclusion $ \rest^{l} \hookrightarrow \rest^{c} .    $ We will often denote the identity morphism $ \idm_\ty : \ty \to \ty $ by the type $\ty $ itself, e.g., $ \tyl \multimap \ty := \idm_{\tyl} \multimap \idm_{\ty} .   $

\begin{example}
We have a morphism $   (\seq{o,o, \seq{o} \multimap o} \multimap o) \to  (\seq{o, \seq{o} \multimap o} \multimap o)     $ in the category of cartesian resource types. The morphism is built out of $  c_{o} : \seq{o} \to \seq{o,o}   $, used in a contravariant way, since the intersection type is in negative position.
\end{example}

\section{Resource Terms}
We now introduce the \emph{structural resource $\lambda$-calculus}, which can be thought of as a \emph{non-uniform} counterpart to the standard simply-typed $\lambda$-calculus in which the structural operations are kept \emph{explicit}. Non-uniformity stems from the application rule, whereas a function can be applied to bags (\textit{i.e.}, lists) of inputs of non-uniform shape and type: \textit{e.g.}, given a function $s $, we apply it to a list of arguments
obtaining $ s \seqdots{t}{1}{k} $, where $ t_1,\dots, t_k $ are arbitrary terms. In this way, our calculus becomes a refinement of both Boudol's \cite{Boudol93} and Ehrhard and Regnier's \cite{EhrhardRegnier08} resource calculi. The main difference is that we replace \emph{multisets} of arguments with \emph{lists}, making the operational semantics completely deterministic. In doing so, we build on several previous proposals \cite{tao:gen, mazza:pol,ol:group} where resource calculi based on lists are considered. The novelty of our approach relies on the treatment of \emph{structural rules}. The calculus is built in such a way as to separate the \emph{structural} (and in particular \emph{exponential}, in the sense of linear logic) side of $ \beta$-reduction from the actual substitution operation, which remains linear.  

\begin{definition}[Structural Resource Terms]
We fix a countable set of variables $  \mathcal{V} $. We define the set of \emph{structural resource terms} ($\rTerms^\spadesuit $) by the following inductive grammar, parametrically over  $ \spadesuit \in \{ l, c \}$
\[        \rTerms^{\spadesuit} \ni s ::= x \in \mathcal{V} \mid \la{x^{f}  } s \mid s  \vec{t}   \qquad    \qquad    \vec{t} \in \oc (\rTerms^{\spadesuit} )        \]
where $ f \in \morph{ \rest^{\spadesuit} }$. Resource terms are considered up to renaming of bound variables. Elements of $\oc(\rTerms^{\spadesuit} )  $ are called \emph{bags} of resource terms. Typing contexts, consisting of sequences of variable type declariations, are defined in Figure \ref{fig:rtctx}. 
\end{definition}
A morphism of typing contexts consists of a sequence of morphisms between intersection types. Typing contexts and their morphisms determine a category that we denote as $ \mathsf{Ctx}^{\spadesuit}$. Terms are considered up to renaming of bound variables. The tensor product of contexts consists of the pairwise concatenation of intersection types assigned to the same variable. Typing rules are given in Figure \ref{fig:quasit}. Typing contexts in multiple-hypothesis rules are always assumed to be joinable, \textit{i.e.}, they  contain the same free variables. The type system admits two kinds of judgments: namely $ \vdash $ and $\vdash_b $ . The former infers the typing of resource terms; the latter that of bags of resource terms. We note that bags of resource terms are typed with intersection types and they are not formally a part of the grammar of resource terms. Given a context of the shape $ x_1 : \seq{},\dots, x : \tyl,\dots , x_n : \seq{}  $ we shall often abuse the language and just denote it by $\seq{\alpha; \vec{f}} $.
 
The main novelty of the structural resource calculus is the typing of abstractions: whenever we bound variables, we are allowed to decorate them with morphisms of intersection types. These morphisms express the (nested) structural rules that we are allowing on (the occurrences of) the considered bound variables. 

\begin{remark}
 The annotation of bound variables in a term like $\lambda x^{f} . s$ intuitively retains the information about the (nested) \emph{structural} operations being performed on the occurrences of the variable $x$. For instance, in the term $s = \la{ x^{c_o}} v \seq{x,x}$, where $c_o : \seq{o} \to \seq{o,o}$, we are performing a \emph{contraction} on the two occurrences of $x$. Thus, if we take a term $t$ of type $o$ and apply it to $s$, it will be \emph{duplicated} along the computation. This means that our reduction semantics should allow for $s \seq{t} \to v \seq{t,t}$. In the cartesian world, any structural operation is allowed; for example, we also have weakenings: $\lambda x^{\mathsf{T}_o} z$. In the linear world, however, the only allowed operations are \emph{permutations}.
\end{remark}

\begin{example}[Typing Some Resource Terms]
Consider the ground morphism $ \alpha = ( \seq{\seq{o} \multimap o } \otimes c_o) : \seq{  \seq{o,o} \multimap o,o       } \to \seq{  \seq{o,o} \multimap o, o, o}  ,$ which performs a contraction on the type $o $ while keeping the arrow type fixed. We can type the term $  \la{x^{\alpha}} x \seq{x,x}  $ as $    \vdash \la{x^{\alpha}} x \seq{x,x } :                 \seq{\seq{o,o} \multimap o, o }  \multimap o .  $
\end{example}

Our system induces a strong correspondence between typed terms and their type derivations, as witnessed by the following result.

\begin{proposition}[Uniqueness of Derivations]
Let $ \pi $ be a derivation of $\gamma \vdash s : \ty $ and $ \pi'$ of  $ \gamma \vdash s : \ty' $. Then $ \ty = \ty'$ and $  \pi = \pi'$.
\end{proposition}

Thanks to this result, we can just work  with typed terms $ \gamma \vdash s : \ty $, without the need of always explicitly referring to the whole type derivation.


\begin{figure*}[t]
\begin{gather*}
\begin{prooftree}
\infer0{  id_o : o \to o   }
\end{prooftree}
 \qquad
\begin{prooftree}
\hypo{    \seq{\alpha, \vec{f}} : \vec{b} \to \tyl                     } \hypo{    f : \ty \to b    }
\infer2{    (\seq{\alpha, \vec{f}}  \multimap f) : (\tyl \multimap \ty) \to (\vec{b} \multimap b)                       }
\end{prooftree}
\\[0.5em] \begin{prooftree}
\hypo{  \alpha \in \Of^{\spadesuit}([m],[n])    }\hypo{  f_{1} : \ty_{\alpha (1)} \to b_1 \dots    f_{m} : \ty_{\alpha (m)} \to b_m       }
\infer2{          \seq{\alpha; f_1, \dots, f_m} : \seqdots{\ty}{1}{n} \to \seq{b_1, \dots, b_m}          }
\end{prooftree} 
\end{gather*}
\caption{Type Morphisms.}
 \hrule
\end{figure*}

\begin{figure*}
\begin{gather*}
\begin{prooftree}
\infer0{  ()   \text{ context }      }
\end{prooftree} 
\qquad
\begin{prooftree} 
\hypo{   \gamma  \text{ context }       } \hypo{ x \text{ fresh variable }, \tyl \in \oc (  \mathcal{RT} )      }
\infer2{     \gamma, x : \tyl \text{ context}       }
\end{prooftree}
\end{gather*}
\caption{Context Formation.}
\label{fig:rtctx}
\hrulefill
\end{figure*}

\begin{figure*}
\begin{gather*}    \begin{prooftree} \infer0{  x_1  : \seq{}, \dots, x_i : \seq{\ty}, \dots ,  x_n : \seq{} \vdash x_i : \ty } \end{prooftree}    \\[1em]  \begin{prooftree}  \hypo{ \gamma, x : \vec{b} \vdash s : \ty}
\hypo{        f : \tyl \to \vec{b}       } 
\infer2{  \gamma \vdash \la{x^{f}} s : \tyl \multimap \ty }     \end{prooftree} \qquad \begin{prooftree} \hypo{  \gamma_0 \vdash s : \tyl  \multimap b  }\hypo{\gamma_1 \vdash_b \vec{t} : \tyl  } \infer2{   \gamma_0 \otimes \gamma_1 \vdash s \vec{t} : b     } \end{prooftree}  
  \\[1em] 
   \begin{prooftree} \hypo{  \gamma_1 \vdash t_1 : \ty_1 } \hypo{\cdots} \hypo{\gamma_k \vdash t_k : \ty_k   }\infer3{\gamma_1 \otimes \cdots \otimes \gamma_k   \vdash_b  \seq{  t_1, \dots, t_k} : \seqdots{\ty}{1}{k}}  \end{prooftree} 
\end{gather*}
\caption{Typing Assignment for Structural Resource Terms.}\label{fig:quasit}
\hrulefill
\end{figure*}

\paragraph*{\textbf{$ \eta$-long forms}}

\begin{figure*}[t]
\scalebox{0.9}{\parbox{1.05\linewidth}{\begin{gather*} 
\begin{prooftree}
\infer0{ x_1 : \seq{}, \dots, x_i : \seq{\overline{\ty }\multimap o } , \dots, x_n : \seq{} \vdash_v x_i : \overline{\ty} \multimap o      }
\end{prooftree} \qquad 
 \begin{prooftree}  \hypo{ \gamma, x_1 : \vec{b}_1 , \dots, x_n : \vec{b}_n \vdash s :  o} \hypo{  f_i : \tyl_i \to \vec{b}_i }
\infer2{  \gamma \vdash \la{(x_1^{f_1} \cdots x_n^{f^n})} s : (\tyl_1 \cdots \tyl_n)  \multimap o }     \end{prooftree}
 \\[1em]
  \begin{prooftree} \hypo{\gamma \vdash_v  x : \overline{\ty} \multimap o } \hypo{ \gamma' \vdash_s \overline{q} : \overline{\ty}    }  \infer2{\gamma \otimes \gamma' \vdash x \overline{q} :  o } \end{prooftree}
\qquad 
    \begin{prooftree} 
\hypo{ \gamma_0 \vdash \la{\overline{x}^{ \overline{f}} }s :  \overline{\ty} \multimap o}  \hypo{\gamma_1 \vdash_s  \overline{t} : \overline{\ty}}
\infer2{ \gamma_0 \otimes \gamma_1 \vdash (\la{\overline{x}^{\overline{f}}}s) \overline{t} : o  }
 \end{prooftree} \\[1em]
  \begin{prooftree}
\hypo{  \gamma_1 \vdash_b \vec{q}_1 : \tyl_1  }\hypo{   \cdots }\hypo{  \gamma_k \vdash_b \vec{q}_n : \tyl_n }
\infer3{   \gamma_1 \otimes \cdots \otimes \gamma_n \vdash_s (\vec{q}_1 \cdots \vec{q}_n) : (\tyl_1 \cdots \tyl_n)                  }
\end{prooftree} \qquad
  \qquad
   \begin{prooftree} \hypo{  \gamma_1 \vdash t_1 : \ty_1 } \hypo{\cdots} \hypo{\gamma_k \vdash t_k : \ty_k   }\infer3{\gamma_1 \otimes \cdots \otimes \gamma_k   \vdash_b  \seq{  t_1, \dots, t_k} : \seqdots{\ty}{1}{k}}  \end{prooftree} 
\end{gather*}}}
\caption{Structural Resource Terms in $ \eta$-long Form.}\label{fig:long} \hrulefill
\end{figure*}

We now introduce a class of structural resource terms, namely the so called $\eta $-long forms.
Focusing on $\eta$-long forms allows for a smoother definition of the rewriting systems we are
going to introduce in the next section. The notion of $\eta$-long term we consider relies on full application: an application $s \vec{t}_1 \cdots \vec{t}_n $ should be of atomic type $o$, 
both when $s$ is a variable and when $s$ is a function. In doing so, we follow~\cite{clair:resc}. 
We first define resource \emph{types} in $ \eta$-long form:
\begin{align*}
\eta(\rest)   \ni \ty   &::= o \mid \overline{\ty} \multimap o \\
\overline{\ty} &::= (\tyl_1 \cdots \tyl_k) \qquad \tyl \in \oc(\eta(\rest))     
\end{align*}
A resource type in $ \eta$-long form can be translated to a standard resource type by setting $  ((\tyl_1 \cdots \tyl_n) \multimap o)^{\ast} = (\tyl_1)^\ast \multimap (  \cdots \multimap (\tyl_n)^\ast\multimap o )   \cdots   )$.  There is also the evident inverse translation, that exploits the fact that any resource type $\ty $ can be written uniquely as a sequence of implications $ \ty = \tyl_1 \multimap ( \cdots \multimap (\tyl_n \multimap o    ) \cdots ) $ for some $  \tyl_i \in \oc (\rest). $

Structural resource \emph{terms} in $ \eta$-long form are those structural resource terms defined by the following inductive grammar:
\begin{align*}
\eta(\Lambda^{\spadesuit}) \ni  s,t &::= x \overline{q} \mid \la{(x_1^{f_1} \cdots x_n^{f_n})} s \mid  (\la{\overline{x}^{\overline{f}}} s )\overline{q} \\
\overline{q} &::= ( \vec{s}_1 \cdots \vec{s}_k )  \qquad \vec{s} \in \oc(\eta(\Lambda^{\spadesuit})) 
\end{align*}
The main difference with the standard grammar is the presence of term sequences $ \overline{q} $, not to be confused with bags. Such a term sequence is in fact a sequence $ \vec{q}_1, \dots, \vec{q}_n $ \emph{of bags}. In the above, $ \overline{f}$ stands for a sequence of type morphisms having the same length as $ \overline{x} $. Typing rules are in Figure \ref{fig:long}. A term in  $\eta $-long form induces a resource term in the evident way. Conversely, any resource term $  \gamma \vdash s : \tyl_1 \multimap (\cdots \multimap (\tyl_n \multimap o  ) \cdots) $ can be easily attributed a $\eta $-long form.

\paragraph*{\textbf{Actions of type morphisms}} Given a type morphism $ f : \ty \to b $ and a typing derivation $  \pi $ of conclusion $  \gamma \vdash s  : a $, it is natural to expect that we can transform $ \pi$ into another derivation $[f]\pi$ of conclusion $\gamma \vdash [f] s  : b     $ for some term $ [f]s$. Something analogous should hold for context morphisms: if we take $ \theta : \delta \to \gamma  $ we should get a new derivation $  \ract{\pi}{\theta}  $ of conclusion $  \delta \vdash \ract{s}{\theta} : \ty$. Intuitively, these transformations consist in applying appropriate structural rules at the level of free and bound variables. Consider, for instance, the following typing \[ x : \seq{  \seq{o} \multimap o }, y : \seq{o}  \vdash x \seq{y} : o  \] 
and the type morphism $(c_{o} \multimap o ) : (  \seq{  o,o  } \multimap o   ) \to (  \seq{o} \multimap o )$. This morphism \emph{replaces} a function that demands \emph{one} input with another that demands \emph{two} inputs. In order to accommodate this request, the bag of inputs must be \emph{changed}, \textit{i.e.}, the variable $ y $ has to be duplicated. In doing so, however, we are obliged to change the typing context, too. The new context cannot be just $ x :  \seq{\seq{  o,o  } \multimap o} , y : \seq{o} $ but rather $ \delta =  x :  \seq{\seq{  o,o  } \multimap o} , y : \seq{o,o}$. Hence, the naive intuition about our transformation is not entirely correct: when we apply $ \theta $ to $ \pi , $ the resulting conclusion will be $  \delta^{ [\nu] } \vdash \ract{s}{\theta} : \ty ,   $ where $ \nu $ is an appropriate ground morphism that depends on $ \pi $ and $ \theta . $ The action of $ \nu $ represents the propagation of the action of $ \theta $ on other variables of the considered term. We shall now give the formal definition of these transformations. Our definition is inspired by the notion of actions of symmetries over linear resource terms, introduced in \cite{tao:gen} and on its generalization presented in \cite{ol:thesis}. 

Given a type derivation $ \pi $ of conclusion $ \gamma \vdash s : \ty $ and morphisms $\theta : \delta \to \gamma $ and $ f : \ty \to b $, we define the \emph{contravariant} and \emph{covariant} actions on $\pi$ induced by $\theta$ and $f$, as follows
\[     \begin{prooftree}
\hypo{\actr{\pi}{\theta} }
\ellipsis{}
{  \delta^{[\nu^{\theta}_s]}  \vdash \actr{s}{\theta} : \ty             }
\end{prooftree}   
\qquad \qquad 
           \begin{prooftree}
\hypo{ [f] \pi }
\ellipsis{}
{  \gamma^{[\mu^f_s]}  \vdash [f]s : b             }
\end{prooftree}                              \]
where $ \mu^f_s $ and $  \nu^{\theta}_s$ are ground morphisms. The actions are defined by induction in Figure \ref{fig:ra} and Figure \ref{fig:la}, respectively. Given $  x_1 : \tyl_1,\dots, x_n : \tyl_n \vdash s : \ty $ and $ \theta : \delta \to (\tyl_1,\dots,\tyl_n) $, we have that $  \nu^{\theta}_s = \nu_1,\dots, \nu_n $ for appropriate ground morphisms $\nu_i $. We denote $  \nu^{\theta}_{x_i,s}  $ the morphism $\nu_i $, that is the one acting on the type of $x_i $. We will often shorten it to $ \nu_{x_i} $, when the remaining information is clear from the context.  We often write just $ \nu, \mu$ keeping the parameters $s,f,\theta$ implicit. In Figure \ref{fig:ra}, we assume that $ \theta  $ is of the shape $ \theta = \seq{\idm; \vec{f}_1} ,\dots , \seq{\idm; \vec{f}_n} $ for some $ n \in \mathbb{N} $. The assumption allows us to decompose the morphism $\theta $ along the multiplicative merging of contexts in the case of multiple-hypothesis rules. We can extend the definition to arbitrary morphisms, simply by setting $    \actr{\pi}{\theta} = \actr{\pi}{  \mathsf{nest}(\theta)}   $ (see Remark \ref{rem:struct}) and $ \nu_s^\theta = \nu^{\mathsf{nest}(\theta)}_s \circ \mathsf{ground}(\theta) $. We have that $ \ract{s}{\idm} = s $ and $ [\idm]s = s  $. We also have compositionality: 

 \begin{proposition}[Compositionality]
 Let $ \gamma \vdash s : \ty $ and $  \theta : \delta \to \gamma, \eta : \delta' \to \delta , f : \ty \to b, g : b \to c$. Then, the following statements hold. 
 \begin{align*}
  &  (\gamma^{[\mu^f_s]})^{[\mu_{[f]s}^g]} \vdash [g]([f]s) : c &  &=&   &\gamma^{[\mu_{s}^{gf}]} \vdash [gf] s : c .&
  \\[0.3em]
   &(\delta'^{[\nu^{\theta}_s]})^{[\nu_{ \actr{s}{\theta} }^{\nu^{\theta}_s \eta}]} \vdash \actr{\actr{s}{ \theta }}{\nu^{\theta}_s \eta} : \ty& &=&  &\delta'^{[ \nu^{\theta \eta}_{s} ] }\vdash \actr{s}{\theta \eta} : \ty.&   
 \end{align*}
 \end{proposition}

\begin{figure*}[t]
\scalebox{0.9}{\parbox{1.05\linewidth}{\begin{align*}
& \left( 
\begin{prooftree}
\infer0{ x : \seq{  \overline{\ty} \multimap o} \vdash_v x : \overline{\ty}   \multimap o}
\end{prooftree}
\right) \{    \seq{\overline{f} \multimap o}            \}
&=&
&\begin{prooftree}
\infer0{ x : \seq{  \overline{b} \multimap o} \vdash_v x : \overline{b}   \multimap o}
\end{prooftree}&
\\[1em]
 &\left(\begin{prooftree}\hypo{ \gamma_1 \vdash_v x : \overline{\ty} \multimap o } \hypo{ \gamma_2 \vdash_s \overline{q} : \overline{\ty}   }  \infer2{  \gamma_1 \otimes \gamma_2 \vdash x \overline{q} : o } \end{prooftree} \right) \{   \eta  \otimes \theta    \} 
 &=& 
&\begin{prooftree}\hypo{\delta_1 \vdash_v x : \overline{b} \multimap o} \hypo{ (\delta_2\thinspace^{[ \nu_2 ]})^{[\mu]}  \vdash_s [\overline{f}](\actr{\bar{q}}{\theta})  : \overline{b}    }  \infer2{ \delta_1 \otimes (\delta_2\thinspace^{ [\nu_2  \mu  ]}) \vdash x [\overline{f}](\actr{\overline{q}}{\theta}) : o } \end{prooftree} &
 \\[1em]  
 &\left(\begin{prooftree}  \hypo{ \gamma, \overline{y} : \overline{b} \vdash s : o} \hypo{      \overline{f} : \overline{\ty} \to \overline{b}    }	
\infer2{  \gamma \vdash \la{\overline{y}^{\overline{f}}} s : \overline{\ty} \multimap o }   \end{prooftree} \right) \{ \theta  \} 
 &=& 
 &\begin{prooftree}  
 \hypo{ \delta^{[ \nu ]}, \overline{y} :  \overline{b}^{[\nu_x]} \vdash \actr{s}{\theta, \overline{b}} : o}
 \hypo{     \nu_x \overline{f} : \overline{\ty} \to  \overline{b}^{[\nu_x]}        }
\infer2{  \delta^{[ \nu ]} \vdash \la{\overline{y}^{ \nu_x  \overline{f } }} \actr{s}{\theta, \overline{b}} : \overline{\ty} \multimap o }     \end{prooftree} & 
 \\[1em]
&\left(\begin{prooftree} 
\hypo{ \gamma_0 \vdash \la{\overline{x}^{\overline{f}} }s :  \overline{\ty} \multimap o}  \hypo{\gamma_1 \vdash \overline{t} :  \overline{\ty}}
\infer2{ \gamma_0 \otimes \gamma_1 \vdash (\la{\overline{x}^{\overline{f}}}s)  \overline{t} : o  }
 \end{prooftree}  \right) \{ \theta_0 \otimes \theta_1  \}
 &=& 
 &\begin{prooftree} 
\hypo{ \delta_0^{ [ \nu_0 ] } \vdash \actr{(\la{\overline{x}^{\overline{f}}}s)}{\theta_0} : \overline{\ty} \multimap o} \hypo{\delta_1^{ [ \nu_1 ]} \vdash_s \actr{\overline{t}}{\theta_1} : \overline{\ty}}
\infer2{ (\delta_0 \otimes \delta_1)^{ [\nu_0 \otimes \nu_1 ]} \vdash \actr{(\la{\overline{x}^{\overline{f}}}s)}{\theta_0} \actr{\overline{t}}{\theta_1} : 0  }
 \end{prooftree} &
  \\[1em] 
 &\left( \begin{prooftree} \hypo{  (\gamma_i \vdash \vec{q}_i : \tyl_i)_{i =1}^n } \infer1{ \bigotimes \gamma_i \vdash_s  (\vec{q}_1 \cdots \vec{q}_n ) : (\tyl_1 \cdots \tyl_n )}  \end{prooftree} \right) \{ \bigotimes \theta_i \} 
    &=& 
  &   \begin{prooftree} \hypo{  (\delta^{[\nu_i]}_i \vdash \actr{\vec{q}_i}{\theta_i} : \tyl_i)_{i =1}^n } \infer1{ (\bigotimes \delta_i)^{ [\bigotimes \nu_i]} \vdash_s  (\actr{\vec{q}_1}{\theta_1} \cdots  \actr{\vec{q}_n}{\theta_n}) : (\tyl_i)}  \end{prooftree} &
  \\[1em] 
 &\left( \begin{prooftree} \hypo{  (\gamma_i \vdash t_i : \ty_i)_{i =1}^k } \infer1{ \bigotimes \gamma_i \vdash_b  \seq{  t_1, \dots, t_k} : \seqdots{\ty}{1}{k}}  \end{prooftree} \right) \{  \bigotimes \theta_i  \} 
    &=& 
  &   \begin{prooftree} \hypo{  (\delta_i^{ [ \nu_i ]} \vdash t_i \{ \theta_i \} : \ty_i)_{i = 1}^k } \infer1{(\bigotimes \delta_i)^{ [\bigotimes \nu_i]}   \vdash_b  \seq{  \actr{t_1}{\theta_1}, \dots, \actr{t_k}{\theta_k}} : \seq{\ty_i}}  \end{prooftree} &
\end{align*}}}
\caption{Contravariant action on resource terms.  We omit the explicit typing of morphisms to improve readability. In the variable rule we assume that $ \overline{f} : \overline{\ty} \to \overline{b}. $ In the application rule for the variable, we assume that the value of $ \eta$ on $ x$ is $ \overline{f} \multimap o .$  }\label{fig:ra} \hrulefill
\end{figure*}

\begin{figure*}[t]
\scalebox{0.9}{\parbox{1.05\linewidth}{\begin{align*}    
& [\overline{f} \multimap o ] \left( \begin{prooftree}
\infer0{  x : \seq{\overline{\ty} \multimap o} \vdash_v x : \overline{\ty} \multimap o       }
\end{prooftree} \right)
& = &
&\begin{prooftree}
\infer0{   x  : \seq{\overline{b} \multimap o} \vdash_v x : \overline{b} \multimap o       }
\end{prooftree}& \\[1em]
 &[id_o]\left(\begin{prooftree} 
\hypo{ \gamma_0 \vdash s :  \overline{\ty} \multimap o}  \hypo{\gamma_1 \vdash \overline{t} :  \overline{\ty}}
\infer2{ \gamma_0 \otimes \gamma_1 \vdash s \overline{t} : o  }
 \end{prooftree}  \right) 
  &=& 
&\begin{prooftree} 
\hypo{ \gamma_0 \vdash  s :  \overline{\ty} \multimap o}  \hypo{\gamma_1 \vdash \overline{t} :  \overline{\ty}}
\infer2{ \gamma_0 \otimes \gamma_1 \vdash s  \overline{t} : o  }
 \end{prooftree}  & \\[1em]
&[ \overline{f} \multimap o]\left(\begin{prooftree}  \hypo{ \gamma, \overline{x} : \overline{c} \vdash s : o} \hypo{  \overline{g} : \overline{\ty} \to \overline{c} }
\infer2{  \gamma \vdash \la{\overline{x}^{\overline{g}}} s : \overline{\ty} \multimap o }   \end{prooftree} \right) 
 &=& 
 &\begin{prooftree}  \hypo{ \gamma, \overline{x} : \overline{b}  \vdash s : o}\hypo{   \overline{gf} : \overline{b} \to \overline{c}     }
\infer2{  \gamma \vdash \la{\overline{x}^{\overline{gf}  }} s : \overline{b} \multimap o }   \end{prooftree} & 
  \\[1em] 
 &[(\seq{\alpha_i;\vec{f}_i})_{i = 1}^n]\left( \begin{prooftree} \hypo{  (\gamma_i \vdash \vec{q}_i : \tyl_i)_{i =1}^n } \infer1{ \bigotimes \gamma_i \vdash_s  (\vec{q}_1 \cdots \vec{q}_n) : (\tyl_1 \cdots \tyl_n)}  \end{prooftree} \right)  
    &=& 
  &  \begin{prooftree} \hypo{  (\gamma^{[\mu_i]}_i \vdash [\seq{\alpha_i; \vec{f}_i}]\vec{q}_i : \vec{b}_i)_{i =1}^n } \infer1{ \bigotimes \gamma^{[ \bigotimes \mu_i]}_i \vdash_s ([\seq{\alpha_1; \vec{f}_1}]\vec{q}_1 \cdots [\seq{\alpha_n; \vec{f}_n}]\vec{q}_n) : (\vec{b}_1 \cdots \vec{b}_n)}  \end{prooftree}&
  \\[1em] 
 &[\seq{\alpha,f_i}_{i= 1}^l]\left( \begin{prooftree} \hypo{  (\gamma_i \vdash t_i : \ty_i)_{i =1}^k } \infer1{ \bigotimes \gamma_i \vdash_b  \seq{  t_1, \dots, t_k} : \seqdots{\ty}{1}{k}}  \end{prooftree} \right) 
    &=& 
  &  \begin{prooftree} \hypo{  (\gamma^{[\mu_j]}_{\alpha(j)} \vdash [f_j] t_{\alpha (j)} : b_j)_{j =1}^l } \infer1{ (\bigotimes \gamma_{j})^{[ (\bigotimes \mu_j )\alpha^\star]} \vdash_b  \seq{  [f_1]t_{\alpha (1)}, \dots, [f_l]t_{\alpha(l)}} : \seqdots{b}{1}{l}}  \end{prooftree} &
\end{align*}}}
\caption{Covariant action on resource terms. We omit the explicit typing of morphisms to improve readability. We assume that $ \overline{f} : \overline{b} \to \overline{\ty}$ . The ground morphism $\alpha^\star$ is the canonical one of type $ \alpha^\star : \bigotimes \gamma_i \to \bigotimes \gamma_{\alpha(j)}$.}\label{fig:la} \hrulefill
\end{figure*}
\section{The Two Rewrite Relations}
We now discuss the reduction semantics of our calculus, which we call \emph{structural reduction}. Structural reduction is defined as the union of two disjoint rewrite relations, called \emph{linear} and \emph{exponential} reduction. While the former mimics $\beta$-reduction, the second should be understood as a process of structural rule elimination. We will show how the two interact, proving that they commute in a certain weak sense. We also prove the strong normalization and confluence of both exponential and linear reduction, from which we obtain the same result for structural reduction as a whole.

\begin{figure*}[t]
\scalebox{0.8}{\parbox{1.05\linewidth}{\begin{align*}    
   \left( \begin{prooftree} \infer0{   x : \seq{\ty}  \vdash x : \ty} \end{prooftree}\right) \{   \seq{t}  / x_i            \}  
& \qquad =  \qquad \gamma \vdash t : \ty                                 
 \\[1em]
\left(\begin{prooftree}  \hypo{ \gamma, x : \tyl, y : \vec{b} \vdash s : b}  \hypo{     f : \tyl \to \vec{b}              }
\infer2{  \gamma, x : \tyl \vdash \la{y^{f}} s : \tyl \multimap b }     \end{prooftree} \right) \{    \vec{t}  / x         \}
 & \qquad = \qquad 
 \begin{prooftree}  \hypo{ (\gamma \otimes \delta)^{[\sigma]} , y :(\vec{b})^{[\sigma']} \vdash \subst{s}{x}{\vec{t}} : b} \hypo{  \sigma' f   }
\infer2{ (\gamma \otimes \delta)^{[\sigma]} \vdash \la{y^{ \sigma' f }} (\subst{s}{x}{\vec{t}}) : \tyl \multimap b }     \end{prooftree}
\\[1em]
\left(    \begin{prooftree} \hypo{  \gamma_0, x : \tyl_0 \vdash s : \tyl  \multimap b  }\hypo{\gamma_1, x : \tyl_1 \vdash \vec{q} : \tyl  } \infer2{   \gamma_0 \otimes \gamma_1, x : \tyl_0 \oplus \tyl_1 \vdash s \vec{q} : b     } \end{prooftree} \right) \{ \vec{t}_0 \oplus \vec{t}_1         \} & \qquad = \qquad 
 \begin{prooftree} \hypo{  (\gamma_0 \otimes \delta_0)^{[\sigma_0]} \vdash \subst{s}{x}{\vec{t}_0} : \tyl  \multimap b  }\hypo{(\gamma_1 \otimes \delta_1)^{[\sigma_1]} \vdash \subst{\vec{q}}{x}{\vec{t}_1} : \tyl  } \infer2{   ((\gamma_0 \otimes \gamma_1) \otimes (\delta_0 \otimes \delta_1))^{[  (\sigma_1 \oplus \sigma_2) \tau]} \vdash \subst{s}{x}{\vec{t}_0} \subst{\vec{q}}{x}{\vec{t}_1} : b     } \end{prooftree} 
 \\[1em]
 \left( \begin{prooftree}
 \hypo{      (\gamma_i, x : \tyl_i \vdash s_i : b_i)_{i =1}^k          }
 \infer1{   \bigotimes \gamma_i, x : \bigoplus \tyl_i \vdash_b \seqdots{s}{1}{k} : \vec{b}  }
 \end{prooftree} \right) \{  \bigoplus \vec{t}_i  / x                 \} & \qquad = \qquad  \begin{prooftree}
 \hypo{      ( (\gamma_i \otimes \delta_i)^{[\sigma_i]} \vdash \subst{s_i}{x}{\vec{t}_i} : b_i)_{i =1}^k          }
 \infer1{   (\bigotimes \gamma_i \otimes \bigotimes \delta_i)^{[ (\bigotimes \sigma_i) \tau ]} \vdash_b \seq{  \subst{s_1}{x}{\vec{t}_1} , \dots, \subst{s_k}{x}{\vec{t}_k} } : \vec{b}      }
 \end{prooftree} 
 \end{align*}}}
\caption{Substitution operation on typed resource terms. In the application and list rules, the list terms $\vec{t}_i $ are of type $\tyl_i $. They correspond to the decomposition of the list $ \vec{t}$, following the multiplicative decomposition of the context. The permutation $ \tau$ is the evident one of type $ \tau : (\bigotimes \gamma_i) \otimes (\bigotimes \delta_i) \to \bigotimes (\gamma_i \otimes \delta_i)  $. }
\hrulefill
\label{fig:linsub}
\end{figure*}

The reduction of structural resource terms relies on \emph{linear substitution}. Given $  \gamma, x : \tyl \vdash s : b  $ and $ \delta \vdash_b \vec{t} : \tyl $, we define a permutation $ \sigma_{s, \vec{t}}$ and a term $ (\gamma \otimes \delta)^{[\sigma_{s, \vec{t}}]} \vdash \subst{s}{x}{\vec{t}} : b $, called the linear substitution of $x $ in $s $ by $ \vec{t} $, by induction, in Figure~\ref{fig:linsub}.

\subsection{Structural Reduction}

We define a relation 
$$
{\to} \subseteq {\Lambda^c  \times (  \morph{ \mathsf{Ctx}^c  } \times \morph{\rest^c } ) \times \Lambda^c},
$$
called \emph{structural reduction} by induction, see Figure \ref{fig:colla}. The relation $\to $ naturally induces a labelled reduction system, simply by equipping $ (  \morph{ \mathsf{Ctx}^c  } \times \morph{\rest^c } )$ with the partial monoid structure induced by composition. We have that if $  (\gamma \vdash s : \ty) \to_{\theta;f}  (\gamma' \vdash s' : \ty')   $ then $   \theta : \gamma \to \gamma'               $ and $ f : \ty' \to \ty .  $ The \emph{exponential} (resp. \emph{linear}) reduction ${\toexp (\text{ resp. }\toli )}$ is defined as the contextual closure of the left-hand-side ( resp. right-hand-side) ground step. We observe that $ {\to} = {\toexp} \cup {\toli}  $. In the following we shall often write $ s \to_{\theta;f} t $, assuming the reduction well-typed and keeping the typing implicit.

\begin{remark} Exponential reduction does not strictly speaking enjoy subject reduction, since types can indeed change. However, the typings of the redex and reduct are related through the type morphisms labeling the reduction step. Typing is not preserved under linear substitution, either. This is due to the possible permutations of variable occurrences. For instance, consider $ s =   ( (\la{y^{\seq{o}}} \la{x^{\seq{\seq{o} \multimap o}}}    xy ) \seq{z}) \seq{z}   $ with the following typing: 
\[ z : \seq{o, \seq{o} \multimap o} \vdash         ((\la{y^{\seq{o}}}       \la{x^{ \seq{ \seq{o} \multimap o}}}   xy ) \seq{z} ) \seq{z}  : o\]
Now, $ s \to z \seq{z} .$ However, the typing context must be permuted, becoming $z : \seq{\seq{o}\multimap o, o} \vdash z\seq{z} : o$. As for exponential reduction, the
evident permutation $ \seq{o, \seq{o}\multimap o} \to \seq{\seq{o}\multimap o, o}$ is recorded as part of the underlying label.
\end{remark} 
 
 \begin{example}
We give an example of structural reduction. Consider the terms from Section~\ref{sec:birdseye}, namely
\begin{align*}
  u &= \la{x^{\seq{\seq{o} \multimap o}^3 }} w \seq{ x \seq{xy}  }  \seq{x \seq{y}};\\
 v &= \la{z^{c_o}} q \seq{z} \seq{z};\\ 
  s &= \la{x^{\seq{\seq{o} \multimap o}^4}} w \seq{ x \seq{ x \seq{y,y} , x \seq{y,y}  }   } \seq{x \seq{y,y}};\\ 
  t &= \la{z^{\seq{o,o}}} q \seq{z,z}.
\end{align*}
 We can normalize $    u\seq{v}^3    $ into $ s \seq{t}^4  $. Since 
$  q : \seq{ (\seq{o} \cdot \seq{o} ) \multimap o } \vdash \la{z^{c_o}} q \seq{z} \seq{z} : \seq{o} \multimap o  )  \to_{\idm; c_o \multimap o} (  q : \seq{ (\seq{o} \cdot \seq{o} ) \multimap o } \vdash \la{z^{\seq{o}^2}} q \seq{z} \seq{z} : \seq{o,o} \multimap o  )   $. Then we have 
\begin{align*}
[c_o \multimap o]  ( w : \seq{(\seq{o}\cdot \seq{o}) \multimap o }  &\vdash s : \seq{o}^3 \multimap o )  \ \to_{\idm; (c_o \oplus \idm) \multimap o} \\
 w : \seq{(\seq{o}\cdot \seq{o}) \multimap o } &\vdash v :                     \seq{o}^4 \multimap o 
 \end{align*}
Hence, we can conclude that $(w : \seq{ (\seq{o} \cdot \seq{o}) \multimap o}, q : \seq{  (\seq{o} \cdot \seq{o}) \multimap o }  \vdash u\seq{v}^3 : o ) \to_{\idm; \idm}   (w : \seq{ (\seq{o} \cdot \seq{o}) \multimap o}, q : \seq{  (\seq{o} \cdot \seq{o}) \multimap o }  \vdash s\seq{t}^4 : o )  .          $
\end{example}

\begin{figure*}
\text{Ground Steps:}

\scalebox{0.9}{\parbox{1.05\linewidth}{\begin{gather*}
\begin{prooftree}
\hypo{           \gamma, \etax : \etaty \vdash s : o \qquad \qquad         \overline{f}  : \overline{b} \to \overline{\ty}   \quad \overline{f} \text{ not identity}               }
\infer1{   (\gamma \vdash \la{\overline{x}^{ \overline{f} }  } s : \overline{b} \multimap o       ) \toexp_{\nu; \nu_x \multimap o} (   \gamma^{[\nu]} \vdash  \la{  \overline{x}^{\overline{b}^{\nu'}}  } \ract{s}{\gamma, \overline{f}}                       )    : \overline{b}^{[\nu']} \multimap o )                    }
\end{prooftree} \\[1em]
\begin{prooftree}
\hypo{  \gamma, \overline{x} : \overline{\ty} \vdash s : o    } \hypo{ \delta\vdash \overline{t} : \overline{b}  }\hypo{  \tau \text{ a permutation} }
\infer3{  ( \gamma \otimes \delta \vdash (\la{\overline{x}^\tau}s) \overline{t} : o ) \toli_{\sigma_{s, [\tau]\overline{t}}; o}    (  (\gamma \otimes \delta)^{[  \sigma_{s; [\tau]\overline{t}} ]} \vdash \subst{s}{\overline{x}}{ [\tau]\overline{t} : o    }    )                     }
\end{prooftree}
\end{gather*}}}
\text{Contextual Closure:}

\scalebox{0.9}{\parbox{1.05\linewidth}{\begin{gather*}
\begin{prooftree}
\hypo{    (\gamma, x : \overline{\ty} \vdash  s :o) \to_{ \theta, \overline{f}; o}  (\gamma', x : \overline{\ty}' \vdash s' : o )  }\hypo{g : \overline{b} \to \overline{\ty}}
\infer2{ (\gamma \vdash   \la{x^{\overline g }} s  : \overline{b} \multimap o ) \to_{ \theta; \overline{b} \multimap o} (\gamma' \vdash \la{x^{\overline{fg}}} s' : \overline{b} \multimap o  ) }
\end{prooftree}
\qquad \\[1em]
\begin{prooftree}
\hypo{  (\gamma  \vdash s : \overline{\ty}\multimap o) \to_{ \theta; \bar{f} \multimap o}  (\gamma'  \vdash s' : \overline{\ty}' \multimap o )}  \hypo{  \delta \vdash_s   \overline{t} : \overline{\ty} }
\infer2{ ( \gamma \otimes \delta \vdash s \overline{t} : o ) \to_{{\theta} \thinspace \otimes \thinspace {\mu^{f}_{\overline{t}}} ; o} (\gamma' \otimes \delta^{[\mu^{f}_{\overline{t}}]}  \vdash  s' [\overline{f}] \overline{t} : o)  }
\end{prooftree} \\[1em]
\begin{prooftree}
\hypo{ (\delta \vdash \overline{t} : \overline{\ty}) \to_{\theta; \overline{f} } (\delta' \vdash  \overline{t}' : \overline{\ty}' )} \hypo{  \gamma \vdash s : \overline{\ty} \multimap o }
\infer2{ (\gamma \otimes \delta \vdash s \overline{t} : o)\to_{  \mu_{ s}^{\overline{f} \multimap o}  \otimes   \theta ; o } ( \gamma^{\mu^{[\overline{f} \multimap o]}_{s}} \otimes \delta' \vdash ([ \overline{f} \multimap o]s ) \overline{t}' : o) }
\end{prooftree} \\[1em]
\begin{prooftree}
\hypo{  (\gamma_1 \vdash s_1 : \ty_1) \to_{\theta_1; f_1 }  (\gamma'_1 \vdash s'_1 : \ty_1')            } \hypo{\cdots}\hypo{ (\gamma_k \vdash s_k : \ty_k) \to_{\theta_k; f_k }  (\gamma'_k \vdash s'_k : \ty'_k)   }  
\infer3{ (\bigotimes \gamma_i \vdash_b \seqdots{s}{1}{k})  \to_{ \bigotimes \theta_i ; \seq{ id; f_1,\dots, f_k } } ( \bigotimes \gamma'_i \vdash_b \seqdots{s'}{1}{k} : \seqdots{\ty'}{1}{k})    }
\end{prooftree} \\[1em]
\begin{prooftree}
\hypo{    (\gamma \vdash_b \vec{q} : \tyl ) \to_{\theta; \seq{\alpha; \vec{f}}} (\gamma' \vdash \vec{q'} : \tyl')      } \hypo{  \delta_1 \vdash_s \overline{q}_1 : \overline{\ty}_1  }\hypo{ \delta_2 \vdash_s  \overline{q}_2 : \overline{\ty}_2  }
\infer3{  (\delta_1 \otimes \gamma \otimes \delta_2   \vdash_s (\overline{q}_1 \cdot \vec{q} \cdot \overline{q}_2) : (\overline{\ty}_1 \cdot \tyl \cdot \overline{\ty}_2) ) \to_{\delta_1 \otimes \theta \otimes \delta_2; \overline{\ty}_1 \cdot \seq{\alpha; \vec{f}} \cdot \delta_2  }  (\delta_1 \otimes \gamma' \otimes \delta_2   \vdash_s (\overline{q}_1 \cdot \vec{q'} \cdot \overline{q}_2 ): (\overline{\ty}_1 \cdot \tyl' \cdot \overline{\ty}_2) )    }
\end{prooftree}
\end{gather*}}}
\caption{Structural reduction.} 
\hrulefill\label{fig:colla}
\end{figure*}

\subsection{Termination and Confluence}

We are interested in proving strong normalization and confluence of the reduction relations introduced in the last section.

We first consider strong normalization of exponential reduction. This will be proved by means of a reducibility argument~\cite{Tait67}. Given that structural reduction is a labelled reduction system, we have to somehow take type morphisms into account. The crucial lemma will then be to prove that if $ \gamma \vdash s : \ty $ and $ \theta : \delta \to \gamma  ,$ then $ \ract{s}{\theta} $ is reducible for the type $\ty$. In order to do so, we need to adjust the classic notion of saturated set of terms. Normally, we would ask for a set to be closed by antireduction: if $ \subst{s}{x}{\vec{t}} \in X $, then $ (\la{x}s)\vec{t} \in X $. We will do the same in our case, taking into account the action of morphisms instead of substitution.

We set $\mathsf{SN} = \{ s \mid \ s \text{ is strongly normalizable for $\toexp$. }  \} .$ A set $X $ of structural resource terms is said to be \emph{saturated} when, given terms $  \gamma, x : \tyl \vdash s : b    $ and $ \delta \vdash \overline{q} : \overline{b} $ with morphism $ \overline{g} : \overline{b} \to \overline{\ty}  $ the following condition holds:  if $ \overline{q} \in \mathsf{SN} $ and for all $ \overline{f}$ s.t.  $ \overline{q} \to_{\theta; \overline{f}} \overline{q'} $ we have that
    $   \ract{s}{\idm;  \overline{gf}}\in X   $, then $   (\la{x^{\overline{g}}} s) \overline{q} \in X .       $ 

\begin{proposition}\label{prop:sn_sat}
$ \mathsf{SN}$ is saturated.
\end{proposition}

Given a resource type $ \ty ,$ we define the \emph{reducibility candidates} for $ \ty  $, $\mathsf{Red}(\ty) \subseteq \rTerms , $ by induction on the structure of $ \ty$, as follows:
\begin{align*}
\mathsf{Red}(o)&=\mathsf{SN};\\
\mathsf{Red}(   \seqdots{\ty}{1}{k}        ) &= \{ \seqdots{t}{1}{k} \mid t_i \in \mathsf{Red}(\ty_i) \} ;\\
     \mathsf{Red}(\overline{\ty} \multimap o) &= \{  s \in \rTerms \mid \text{ for all } \overline{q} \in \mathsf{Red}(\overline{\ty}), s \overline{q} \in \mathsf{Red}(o)\}   ;\\
   \mathsf{Red}((\tyl_1 \cdots \tyl_n)) &= \{  (\vec{t}_1 \cdots \vec{t}_n)  \mid \vec{t}_i \in \mathsf{Red}(\tyl_i)       \}                              
\end{align*} 

\begin{lemma}
We have that $  \mathsf{Red}(\ty) \subseteq \mathsf{SN} .  $
\end{lemma}

\begin{lemma} 
 Let $ \gamma \vdash s : \ty $ and $ \theta : \delta \to \gamma, f : \ty \to b . $ Then $    \ract{s}{\theta} \in \mathsf{Red}(\ty)    $ and $ [f]s \in \mathsf{Red}(b) . $ 
\end{lemma}
\begin{proof}
By mutual induction on the structure of $ s$, exploiting the saturation property in the abstraction cases.
\end{proof}

\begin{theorem}[Adequacy]
If $\gamma \vdash s : \ty $ then $ s \in \mathsf{Red}(\ty). $
\end{theorem}
\begin{proof}
Corollary of the former lemma, by taking $  \theta = \idm_{\gamma}$.
\end{proof}

\begin{theorem}
Exponential reduction is strongly normalizing.
\end{theorem}

Strong normalization of linear reduction can be proved in a combinatorial way, since the size of derivation strictly \emph{decreases} under linear reduction.

\begin{proposition}
If $ s \toli^\ast s' $ then $ \size{s'} < \size{s} $. As a consequence, linear reduction is strongly normalizing.
\end{proposition}

\begin{remark}
Normal forms for the exponential reduction are a special kind of linear resource terms, the \emph{planar} ones. Normal forms for linear reduction, instead, are not linear in general, since non-linear redexes are never fired. Normal forms for general structural reduction are exactly the planar terms in $\beta $-normal form.
\end{remark}

\begin{remark}
Exponential steps can generate linear redexes. For instance, $    (\la{x^{c_o}}w \seq{x}\seq{x}) \seq{y}  $ is a linear normal form; however, if we perform a step of exponential reduction we get the term $ (\la{x^{\seq{o,o}}}  w \seq{x} \seq{x}) \seq{y,y} $ that is now a linear redex. Linear steps too can generate exponential redexes. This is due to the permutations generated by substitution. If you take $ s = \la{z^{\seq{o, \seq{o}\multimap o}}}   ( \la{x}\la{y} y \seq{x} ) \seq{z,z}    $, the term $s $ is an exponential normal form. However, by performing the leftmost linear reduction to its normal form, we obtain the term $  t = \la{ z^{\sigma} } z \seq{z} $, since the substitution forces the permutation of the two occurrences of $z  $. Now $t $ is not an exponential normal form anymore.
Hence, the two reductions do not strictly commute with each other. 
\end{remark}
\begin{figure*}[t]
\scalebox{0.9}{\parbox{1.05\linewidth}{\begin{gather*}
\begin{prooftree}
\infer0{(x : \seq{\ty} \vdash x : \ty  ) \cong_{\seq{\ty}}  (x : \seq{\ty} \vdash x : \ty  )} 
\end{prooftree}
  \qquad
  \begin{prooftree}
  \hypo{  \sigma : \gamma \to \gamma^{[\sigma]}  \text{ is a permutation}              }
  \infer1{ (   \gamma \vdash s : \ty    ) \cong_{\sigma}  (  \gamma^{[\sigma]} \vdash s : \ty ) }
  \end{prooftree} 
  \\[1em]
  \begin{prooftree}
  \hypo{  (\gamma, x : \tyl \vdash s : \ty) \cong_{\theta, \sigma } (\gamma', x : \tyl' \vdash s' : \ty')     } \hypo{ f : \vec{b} \to \tyl }
  \infer2{   (\gamma  \vdash  \la{x^f} s : \vec{b} \multimap \ty) \cong_{\theta } (\gamma'  \vdash \la{ x^{\sigma f} } s' :  \vec{b} \multimap \ty)    }
  \end{prooftree}
  \\[1em]
  \begin{prooftree}
  \hypo{  ( \gamma_0 \vdash s : \tyl \multimap \ty  ) \cong_{\theta_0} ( \gamma'_0 \vdash s' : \tyl \multimap \ty )                }\hypo{   (\gamma_1 \vdash \vec{t} : \tyl) \cong_{\theta_1} (\gamma'_1 \vdash \vec{t}' : \tyl)  }
  \infer2{ (\gamma_0 \otimes \gamma_1  \vdash s\vec{t} : \ty  ) \cong_{\theta_0 \otimes \theta_1}   (\gamma'_0 \otimes \gamma'_1  \vdash s'\vec{t}' : \ty  )}
  \end{prooftree}
  \\[1em]
  \begin{prooftree}
  \hypo{        (  (\gamma_i \vdash t_i : \ty_i) \cong_{\theta_i}  (\gamma' \vdash t'_i : \ty_i ) )_{i = 1}^k                                 }
\infer1{     (\bigotimes \gamma_i \vdash \seqdots{t}{1}{k} : \seqdots{\ty}{1}{k} ) \cong_{\bigotimes\theta_i}      (\bigotimes \gamma'_i \vdash \seqdots{t'}{1}{k} : \seqdots{\ty}{1}{k} )                        }  
  \end{prooftree}
\end{gather*}}}
\caption{Isomorphism of Resource Terms.}
\hrulefill
\label{fig:iso}
\end{figure*}

We are now going to prove a form of weak commutation, mediated by an isomorphism. We define $ {\cong} \subseteq {( \rTerms \times \morph{\mathsf{Ctx}} \times \rTerms  )}  $ called the \emph{isomorphism relation} between resource terms, by induction in Figure \ref{fig:iso}. We note that if $ s \cong_{\theta} s'  $ then $ \theta$ is a permutation. We write $s\cong s'$ iff
$s\cong_\theta s'$ for some $\theta$.

\begin{lemma}
Let $ s\cong s' $ and $   s \toexp p   $ (resp. $ s \toli p$). Then there exists a term $ p' $ s.t. $ s' \toexp p'  $ (resp. $s' \toli p' $) and $ p \cong p' $.
\end{lemma}

\begin{theorem}[Commutation]\label{theo:commutation}
Let $  s \toli^{\ast} t \toexp^\ast t'  $ Then there exist terms $  u, u' $ s.t. $ s \toexp^\ast u \toli^\ast u'  $ with $ t' \cong u'  $. 
\end{theorem}
From these commutation results, we can lift strong normalization of $\toexp $ and $\toli $ to the whole structural reduction:
\begin{theorem}\label{th:snstruct}
Structural reduction is strongly normalizing.
\end{theorem}
\begin{proof}
By Theorem~\ref{theo:commutation}, given an infinite reduction sequence for structural reduction, we could build an infinite one for either exponential or linear reductions, which are instead known to be strongly normalizing.
\end{proof}

Let us now consider confluence. To prove that, we employ a standard method: first we prove local confluence, from which we infer general confluence as a corollary of Newman's Lemma, exploiting the fact that structural reduction terminates.
\begin{theorem}\label{th:confluence}
	Structural reduction is confluent.
\end{theorem}
 

\section{On Approximations, Collapse, and all That}
\subsection{Linearizing Type Systems}

The structural resource $\lambda$-calculus can be seen as a calculus of approximations of ordinary $\lambda $-terms, generalizing what happens in the Taylor expansion \cite{EhrhardRegnier08, barbaro:tay}. We define an \emph{approximation relation} $ {\lhd} \subseteq { \rTerms^c \times \Lambda } $ by induction in Figure \ref{fig:approx}. In general, approximations are \emph{not} linear. However, we can \emph{compute} a linear approximation starting from a non-linear one, just by normalizing through exponential reduction. We shall use this fact to get linearizations of every strongly normalizing $\lambda $-terms.

In this section we will assume that a typed term $\Gamma \vdash M : A $ is in $ \eta$-long form. We remark that any term, typed with either simple or intersection types, has a $ \eta$-long form.

\begin{figure*}
\begin{gather*}
\begin{prooftree}
\infer0{  x \lhd x               }
\end{prooftree} \qquad \begin{prooftree}
\hypo{ s \lhd M  }
\infer1{ \la{x^f} s \lhd \la{x} M  }
\end{prooftree} \qquad
\begin{prooftree}
\hypo{  s \lhd M    } \hypo{ \vec{t} \lhd N }
\infer2{   s \vec{t} \lhd MN                      }
\end{prooftree}
\qquad
\begin{prooftree}
\hypo{ t_1 \lhd M \cdots t_k \lhd M  }
\infer1{  \seqdots{t}{1}{k} \lhd M       }
\end{prooftree}
\end{gather*}
  \caption{Approximation Relation for Ordinary $\lambda $-terms.}
  \label{fig:approx}
\hrulefill
\end{figure*}

\begin{figure}
\[  \mathsf{Idem} \ni A ::= o \mid \tilde{A} \Rightarrow B \qquad \tilde{A} := A_1 \cap \cdots \cap A_k  \ (k \neq 0)                  \]
\begin{gather*}
\begin{prooftree}
\hypo{ A \in \tilde{A}_i   }
\infer1{ x_1 : \tilde{A}_1 ,\dots, x_i : \tilde{A}_i, \dots, x_n : \tilde{A}_n \vdash x_i : A               }
\end{prooftree} \qquad 
\begin{prooftree}
\hypo{  \Gamma, x : \tilde{A} \vdash M : B }
\infer1{       \Gamma \vdash \la{x} M  : \tilde{A} \Rightarrow B         }
\end{prooftree}
\\[1em]
\begin{prooftree}
\hypo{\Gamma \vdash M : \tilde{A} \Rightarrow B} \hypo{  \Gamma \vdash_\cap N : \tilde{A} }
\infer2{   \Gamma \vdash MN : B } 
\end{prooftree}
\qquad
\begin{prooftree}
\hypo{    (\Gamma \vdash M : A_i)_{i=1}^k       }
\infer1{      \Gamma \vdash_\cap M :  A_1 \cap \cdots \cap A_k   }
\end{prooftree}
\end{gather*}
\caption{Idempotent Intersection Type Assignment.}
\label{fig:idem}
\hrulefill
\end{figure}

\paragraph*{Qualitative Fragment and Uniformity}

We can define a fragment of the structural resource calculus into which we shall \emph{embed} every strongly normalizing $\lambda $-term. First, we formally define a \emph{coherence relation } on resource terms $ s \coh s' $, that is a direct generalization of coherence as defined in \cite{EhrhardRegnier08}:
\begin{gather*}  \begin{prooftree}
\infer0{ x \coh x  }
\end{prooftree}      \qquad \begin{prooftree}
\hypo{  s \coh s'  }
\infer1{   \la{x^f } s \coh \la{x^g} s' } 
\end{prooftree}     
\qquad 
\begin{prooftree}
\hypo{ s \coh s'  } \hypo{\vec{t} \coh \vec{t'}}
\infer2{ s\vec{t }  \coh s' \vec{t'}}
\end{prooftree} 
\\[0.5em]
\begin{prooftree}
\hypo{    t_i \coh t'_j \quad \text{ for all } i \in [n],  j \in[m]            }
\infer1{       \seqdots{t}{1}{n} \coh \seqdots{t'}{1}{m}                        }
\end{prooftree}                  
\end{gather*} 
Given an ordinary $ \lambda$-term $ M$, we have that $ s,s'\lhd M$ iff $  s \coh s'$. If $ s \coh s$, we say that $ s $ is \emph{uniform}. We remark that if we have an exponential step $ s \toexp s' $, then $s \coh s' $, while substitution notoriously breaks coherence. We extend the notion of coherence and uniformity to resource type, in the natural way:
\[                  \begin{prooftree}
\infer0{  o \coh o  }
\end{prooftree}     \qquad \begin{prooftree}  \hypo{ \tyl \coh \vec{b} } \hypo{  \ty \coh b }\infer2{  \tyl \multimap \ty \coh \vec{b} \multimap b } \end{prooftree} \qquad \begin{prooftree}
\hypo{    \ty_i \coh b_j \quad \text{ for all } i \in [n],  j \in[m]            }
\infer1{       \seqdots{\ty}{1}{n} \coh \seqdots{b}{1}{m}                        }
\end{prooftree}    \]
 We define the \emph{uniform fragment} of the cartesian resource calculus by induction as follows:
\[   \Lambda^c_{\mathsf{qual}} \ni s ::= x \mid \la{x^\alpha} s \mid s \seq{t_1,\dots, t_k} \qquad (\text{ where } t_i \coh t_j, \forall i, j \in [k], k \neq 0 )    \]
Where $ \alpha $ is a ground morphism. A uniform term is, in particular, uniform in the sense of the coherence relation.  We say that a uniform term is \emph{strongly uniform} when every bag in his body contains copies of the same resource term. We can give a formal definition of strongly uniform terms with the following grammar:
\[   \Lambda^c_{\mathsf{qual}} \ni s ::= x \mid \la{x^\alpha} s \mid s \seq{t_1,\dots, t_k} \qquad (\text{ where } t_i = t_j, \forall i, j \in [k], k \neq 0 )    \]
We extend these notions to types in the natrual way. A typed term $\gamma \vdash s : \ty $ is uniform whenver $ s$ is uniform. Given a strongly uniform term $ s $ with typing $ \gamma \vdash s : \ty$, we call it \emph{qualitative} if $ \gamma$ is strongly uniform and the bags occurring in $s$ and $ \ty$ are singletons. In what follows, we give embeddings of strongly normalizing $ \lambda$-terms into cartesian resource $\lambda$-terms which targets the uniform fragment.

\paragraph*{Simple types}\label{sec:simply} We can easily embed the simply-typed $\lambda$-calculus into the structural resource calculus. We set $ \intt{\ast} = o$ and $\intt{A \Rightarrow B} = \seq{\intt{A}} \multimap \intt{B} $. Given $ n \in \mathbb{N} $ and $A$ a simple type, we define $  \mathsf{cart}_{A}^n : \seq{\intt{A}} \to \seq{\intt{A}}^n  $ as follows: 
\[       
\mathsf{cart}_{A}^0 = \mathsf{T}_{\intt{A}} 
\qquad            
\mathsf{cart}_{A}^{n+1}  = c_{\intt{A}}^{n}                               
\]
Given a simply-typed term $ x_1 : A_1, \dots, x_n : A_n \vdash M  : A$, we associate to it a resource term $ x_1 :  \seq{\intt{ A_1}^{n_1^M}}, \dots, x_n : \seq{\intt{A_n}^{n_n^M}} \vdash \qt{M}^{\Gamma}_A : \intt{A} $ with $ n_i ^M\in \mathbb{N} $, called its \emph{coarse approximation}, by induction, as shown in Figure \ref{fig:coarse}. We often keep the labels $\Gamma, A $ implicit, wiriting just $  \qt{M} $, when the typing can be reconstruced by the context.

\begin{figure*}[t]
\scalebox{0.9}{\parbox{1.05\linewidth}{\begin{gather*}
 \qt{ \begin{prooftree} \infer0{x_1 : A_1, \dots, x_i : A_i, \dots, x_n : A_n \vdash x_i : A_i  }    \end{prooftree} } \quad   = \quad \begin{prooftree} \infer0{x_1 : \seq{}, \dots, x_i : \seq{\intt{A_i}}, \dots, x_n : \seq{} \vdash x_i : \intt{A_i}  }    \end{prooftree} 
       \\[1em]
   \qt{\begin{prooftree}  \hypo{\Gamma, x : A \vdash M : B}\infer1{\Gamma \vdash \la{x^A} M : A \Rightarrow B }    \end{prooftree}} \quad = \quad 
   \begin{prooftree}  \hypo{\intt{\Gamma}^{\overline{n}}, x : \seq{\intt{A}}^{m} \vdash \qt{M}_{B}^{\Gamma,A} : \intt{B}} \hypo{  \mathsf{cart}^m_{A} : \seq{\intt A} \to \seq{\intt{A}}^m    }
   \infer2{\intt{\Gamma}^{\overline{n}} \vdash \la{x^{\mathsf{cart}^m_{A}}} \qt{M}_B^{\Gamma,A} : \seq{\intt{A}} \multimap \intt{B} }   \end{prooftree}
    \\[1em]
     \qt{\begin{prooftree} \hypo{\Gamma \vdash M : A \Rightarrow B} \hypo{\Gamma \vdash N : A}\infer2{\Gamma \vdash MN : B}   \end{prooftree}} 
     \quad = \quad  
      \begin{prooftree} \hypo{\intt{\Gamma}^{\overline{n}} \vdash \qt{M}_{A \Rightarrow B}^\Gamma : \seq{\intt{A}} \multimap \intt{B}} \hypo{\intt{\Gamma}^{\overline{m}} \vdash \qt{N}_A^\Gamma : \intt{A}}\infer2{\intt{\Gamma}^{\overline n} \otimes \intt{\Gamma}^{\overline m} \vdash \qt{M}\seq{\qt{N}} : \intt{B}}   \end{prooftree} 
\end{gather*}}}
\caption{Translation of simply typed terms into cartesian resource terms.}
\hrulefill
\label{fig:coarse} 
\end{figure*}

\begin{figure*}[t]
\scalebox{0.9}{\parbox{1.05\linewidth}{\begin{gather*}
 \qt{ \begin{prooftree} \infer0{x_1 : A_1, \dots, x_i : A_i, \dots, x_n : A_n \vdash x_i : A_i  }    \end{prooftree} } \quad   = \quad \begin{prooftree} \infer0{x_1 : \seq{}, \dots, x_i : \seq{\intt{A_i}}, \dots, x_n : \seq{} \vdash x_i : \intt{A_i}  }    \end{prooftree} 
       \\[1em]
   \qt{\begin{prooftree}  \hypo{\Gamma, x : (\bigcap A_i) \vdash M : B}\infer1{\Gamma \vdash \la{x} M : (\bigcap A_i) \Rightarrow B }    \end{prooftree}} \quad = \quad 
   \begin{prooftree}  \hypo{\intt{\Gamma}^{\overline{n}}, x : \seq{\intt{A_{\sigma(i)}}^{n_{\sigma(i)}}} \vdash \qt{M}_{B}^{\Gamma,\bigcap A_i} : \intt{B}} \hypo{  \mathsf{cart}^{n_{i}}_{A_i} : \seq{\intt A_i} \to \seq{\intt{A_i}}^{n_i}    }
   \infer2{\intt{\Gamma}^{\overline{n}} \vdash \la{x^{ \bigoplus \mathsf{cart}^{n_{\sigma(i)}}_{A_{\sigma(i)}} } } \qt{M}_{B}^{\Gamma, \bigcap A_i} : (\bigcap \intt{A_{\sigma(i)}}) \multimap \intt{B} }   \end{prooftree}
    \\[1em]
     \qt{\begin{prooftree} \hypo{\Gamma \vdash M : \tilde{A} \Rightarrow B} \hypo{\Gamma \vdash N : \tilde{A}}\infer2{\Gamma \vdash MN : B}   \end{prooftree}} 
     \quad = \quad  
      \begin{prooftree} \hypo{\intt{\Gamma}^{\overline{n}} \vdash \qt{M}_{\tilde{A} \Rightarrow B}^\Gamma : \intt{\tilde{A}} \multimap \intt{B}} \hypo{\intt{\Gamma}^{\overline{m}} \vdash \qt{N}_{\tilde{A}}^\Gamma : \intt{\tilde{A}}}\infer2{\intt{\Gamma}^{\overline n} \otimes \intt{\Gamma}^{\overline m} \vdash \qt{M}^\Gamma_{\tilde{A} \Rightarrow B}\qt{N}_{\tilde{A}}^\Gamma : \intt{B}}   \end{prooftree} 
          \\[1em]
\qt{ \begin{prooftree}  \hypo{ \Gamma \vdash M : A_1 \cdots \Gamma \vdash M : A_k}\infer1{ \Gamma \vdash M : (A_1 \cap \dots \cap A_k)}  \end{prooftree}} = \begin{prooftree} \hypo{\intt{\Gamma}^{\overline{n}_1} \vdash \qt{M}_{A_1}^\Gamma : \intt{A_1}   \dots\intt{\Gamma}^{\overline{n}_k} \vdash \qt{M}_{A_k}^\Gamma : \intt{A_k} }
\infer1{      \intt{\Gamma}^{\overline{n}_1} \otimes \cdots \otimes \intt{\Gamma}^{\overline{n}_k} \vdash \seq{\qt{M}_{A_1}^\Gamma,\dots, \qt{M}_{A_k}^\Gamma} : \intt{A_1} \cap \cdots \cap \intt{A_k}           }
  \end{prooftree}
\end{gather*}}}
\caption{Translation of idempotent intersection type derivations into cartesian resource terms.}
\hrulefill
\label{fig:coarseint} 
\end{figure*}

We prove that the embedding preserves $\beta $-reduction, \emph{factorizing it} into exponential and linear steps:

\begin{theorem}
Let $ M \to_\beta N $ then there exists a  cartesian resource term $ t$  s.t. $   \qt{M} \toexp t \toli \qt{N} $.
\end{theorem}
From the former theorem and strong normalization of structural reduction (Theorem \ref{th:snstruct}), we obtain strong normalization of simply-typed terms as a corollary.
We observe that if $M$ is simply-typed, then $ \qt{M} $ is a strongly uniform resource term. From this, we obtain the following characterization: 
\begin{proposition}
$  \gamma \vdash s : \ty$ is qualitative iff there exists a simply typed term $  \Gamma \vdash M : A $ s.t. $  \qt{\Gamma \vdash M : A} = \gamma \vdash s : \ty $.
\end{proposition}

\paragraph*{Idempotent Intersection Types}\label{sec:idem} We shall now show that every strongly normalizing  $\lambda$-term can be represented as a cartesian resource term, by exploiting a classic result about idempotent intersection types. We  consider the system of idempotent intersection types in Figure \ref{fig:idem}, which corresponds to the one from \cite{bakel:strict}. Intersection types can be seen as a finite set of types.  The system enjoys both subject reduction and expansion, from which one can prove that any strongly normalizing $ \lambda$-term can be typed in it (see, e.g., \cite{kri:lambda}). We fix an enumeration of types as the bijection $\emph{enum} : \mathsf{idem}\to\mathbb{N}$. Then we set 
\begin{gather*}\intt{\ast} = o \qquad \intt{ \tilde{A} \Rightarrow B} =  \intt{\tilde{A}} \multimap \intt{B}  \\[0.3em] \intt{ A_1 \cap \cdots \cap A_k} = \seq{\intt{A_{\sigma (1)}},\dots, \intt{A_{\sigma(k)}}} 
\end{gather*} 
where $\sigma $ is the unique permutation s.t. $ \sigma(i) \leq \sigma(j) $ whenever $ \mathsf{enum}(A_i)\leq \mathsf{enum}(A_j)$. Given a typing $ x_1 : \tilde{A}_1, \dots, x_n : \tilde{A}_n \vdash M  : A$ we associate to it a resource term $ x_1 :  \intt{ \tilde{A}_1}^{\vec{n}_1^M}, \dots, x_n : \intt{\tilde{A}_n}^{\vec{n}_n^M} \vdash \qt{M}_A^\Gamma : \intt{A} $ with $ \vec{n}_i^M $ being list of integers of the same length as  $ \intt{\tilde{A}_i}$, called its \emph{coarse approximation}, defined in Figure \ref{fig:coarseint}. We often keep the labels $\Gamma, A $ implicit, wiriting just $  \qt{M} $, when the typing can be reconstruced by the context.

The embedding, again, preserves $ \beta$-reduction. We then get strong normalization of intersection typed terms as a corollary. 

We observe that if $M $ is typed in the idempotent system, we get that $\qt{M} $ is a uniform term. However, $ \qt{M}$ is not strongly uniform as in the simply-typed case. We then obtain the following characterization: 
\begin{proposition}
$\gamma \vdash s : \ty $ is a uniform term iff there exists an idempotent intersection typing  $ \Gamma \vdash M : A $ s.t. $ \qt{\Gamma \vdash M : A} = \gamma \vdash s : \ty   $.
\end{proposition}

\subsection{The Extensional Collapse, Operationally} Our rewriting system allows to explicitly connect linear and non-linear approximations of ordinary $\lambda $-terms. Another way of establishing the connection has been pursued in the context of the \emph{relational semantics} of linear logic by, e.g., Ehrhard \cite{er:collapse}. 
It is well-established that, in this setting, the interpretation of a $ \lambda$-term can be presented by means of \emph{multitypes}, whereas the intersection connective is seen as a multiset $ [A_1,\dots, A_k] $. In particular, we can present both the traditional relational semantics and the Scott semantics in this way (cfr. \cite{er:collapse}). Given an untyped term $ M$ with $ \fv{M} \subseteq \vec{x} $, its denotation is a monotonic relation ${\sem{M}_{\clubsuit}} \subseteq {    D^{\length{\vec{x}}} \times D   }  $ defined as
\[  \sem{M}_{\clubsuit} = \{     (\Gamma, A) \mid \Gamma \vdash_{\clubsuit}  M : A      \}  \]
for $ \clubsuit \in \{ \mathsf{scott}, \mathsf{rel}  \}  $ and $ D$ is (the underlying carrier of) an appropriate preorder on multitypes. The judgment $ \Gamma \vdash_{\clubsuit} M : A $ is obtained through an intersection type system, that is idempotent (resp. non-idempotent) in the Scott semantics (resp. relational semantics). The \emph{extensional collapse of non-idempotent intersection types} is the equality 
\[  \sem{M}_{\mathsf{scott}} = \downarrow\sem{M}_{\mathsf{rel}} := \{ (\Gamma,A) \mid \text{ there exists } (\Gamma',B) \in \sem{M}_{\mathsf{rel}}\] \[  \text{ s.t. } \Gamma \leq_{\mathsf{scott}} \Gamma', B \leq_{\mathsf{scott}} A   \} \]
 where $ A \leq_{\mathsf{scott}} B  $ is a preorder on intersection types generated by the rules: $ [A]\leq [A,A] , \ [A_1, A_2] \leq [A_i] $ and $ [A_1, \dots, A_k] \leq [] $. The preoder $ \leq_{\mathsf{scott}} $ is the standard one used to obtain the idempotency of the intersection type. The result has been proved by Ehrhard semantically \cite{er:collapse}. We can here give an operational proof. If we endow the structural resource $\lambda$-calculus with the following subtyping rule:
\[           \begin{prooftree} 
\hypo{   x_1 : \tyl_1,\dots, x_n : \tyl_n \vdash s : \ty       } \hypo{ f_i : \vec{b}_i \to \tyl_i    }
\infer2{    x_1^{f_1} : \vec{b}_1,\dots, x_n^{f_n} : \vec{b}_n  \vdash_{\mathsf{sub}} s : \ty             }
\end{prooftree}   \] 
This subtyping rule is needed in order to capture the suptying on context induced by the preoder $  \leq_{\mathsf{scott}}  $. We can characterize both forms of semantics by exploiting our calculus. We defined
   the approximation relation as 
\begin{equation}\label{eq:rel}
\sem{M}_{\clubsuit} = \{ ( \underline{\gamma}, \underline{\ty}      ) \mid \gamma \vdash_{\spadesuit}  s : \ty \text{ for some } s \lhd M \}  
\end{equation}
where $ \spadesuit = c   $ (resp. $ l$) if $ \clubsuit = \mathsf{scott} $ (resp. $\mathsf{rel} $) and $ \underline{\ty} $ is obteined from $\ty $ by replacing any list with the corresponding multiset. We can  then rephrase Ehrhard's result into our setting, obtaining an operational and proof-relevant version of the collapse, bridging categorical constructions and our rewriting system. We can perform the collapse for both typed and untyped $\lambda $-calculi. Given a $ \lambda$-term $M $, we set 
\[ \mathsf{Appr}(M)_{\spadesuit}(\gamma;\ty) = \{ s \in \rTerms \mid s \lhd M  \text{ and } \gamma \vdash_{\spadesuit} s : \ty\}.     
\]
Given a typed resource term $ s$, we denote as $\mathsf{nf}^e(s) $  its normal form for the exponential reduction. Thanks to the normalization and confluence result for our exponential reduction, we have that the map $ s\mapsto \mathsf{nf}^e(s) $ defines the function:
\begin{equation*}\label{eq:collapse}
\mathsf{Appr}(M)_{c}^\eta(\gamma;\ty) \to \sum_{\delta,b} \mathsf{Appr}(M)_l^\eta (\delta;b) \times \mathsf{Ctx}^c(\gamma, \delta) \times \mathsf{\rest}^c(b,a)                              
\end{equation*} 
where $ \mathsf{Appr}(M)_{\spadesuit}^\eta(\gamma;\ty)$ denotes the set of $ \eta$-long forms of approximations of $ M$. We call this function the \emph{proof-relevant collapse} of linear intersection types into cartesian ones. From this, we can give an alternative proof of the semantic collapse.

\begin{theorem}
Let $ M $ be a $\lambda $-term. We have that
\[ { \sem{M}_{\mathsf{scott}}} =    { \downarrow\sem{M}_{\mathsf{rel}}}.\]
\end{theorem}
\begin{proof}
($ \subseteq$) Let $  (\Gamma, A)  \in \sem{M}_{\mathsf{scott}} .   $ By $ (\ref{eq:rel}) $, there exists $ \gamma \vdash_c s : \ty $ with $\Gamma = \underline{\gamma} , A = \underline{\ty} .$ We consider a normalization rewriting $   (\gamma \vdash_{l} s : \ty ) \to_{\theta; f} ( \gamma' \vdash_l \mathsf{nf}^e (s) : \ty')        $. Now, again by $ (\ref{eq:rel}) $ we have that $ (\underline{\gamma}, \underline{\ty}) \in \sem{M}_{\mathsf{rel}}  $. Since $ \theta : \gamma \to \gamma' $ and $ f : \ty' \to \ty $ we can infer that $  \underline{\gamma} \leq_{\mathsf{scott}} \underline{\gamma'} $ and $ \underline{\ty'} \leq_{\mathsf{scott}} \underline{\ty} $. Then we can conclude that $  (\Gamma, A) \in   \downarrow\sem{M}_{\mathsf{rel}}.$

($\supseteq$) Let $  (\Gamma, A)  \in \sem{M}_{\mathsf{rel}} $. By $ (\ref{eq:rel}) $, there exists $ \gamma \vdash_l s : \ty $ with $\Gamma = \underline{\gamma} , A = \underline{\ty} .$ Since $ \gamma \vdash_l s : \ty$, in particular $ \gamma \vdash_c s : \ty$. Then, by definition, $  (\underline{\gamma}, \underline{\ty} ) = (\Gamma, A) \in \sem{M}_{\mathsf{scott}}  $.
\end{proof}
\section{Related Work}
Linearization --- seen as a program transformation aiming at bringing $\lambda$-terms into a linear form --- was first considered by Kfoury~\cite{Kfoury00}, whose notion of linearization does
not eliminate structural rules \emph{on demand} but somehow \emph{blindly}, resulting in a notion of reduction which is neither strongly normalizing nor confluent. This contrasts with the structural resource $\lambda$-calculus, in which instead linearization and $\beta$-reduction commute (cfr. Theorem ~\ref{theo:commutation}).
Alves and Florido~\cite{AlvesFlorido03,AlvesFlorido05} consider a \emph{weakly} linear fragment of the $\lambda$-calculus and a transformation of all terms of the ordinary $\lambda$-calculus into this fragment which preserves normal forms; the transformation, however, does not occur internally to the calculus, being based on Levy's legal paths. More recently, Alves and Ventura~\cite{AlvesVentura22} revisit weakly linear terms by characterizing them in terms of non-idempotent intersection types, this way simplifying the translation. Non-idempotent typing being itself a quantitative system, the bottleneck of the translation is somehow transferred to typing.

The work most closely related to ours is certainly the recent one by Pautasso and Ronchi~\cite{PautassoRonchi23}. They present a unification algorithm that, given a simply typed $ \lambda$-term $M$, produces a typing for $ M$ in a restricted (uniform) non-idempotent intersection type system. We obtain a similar result as a corollary of our work, as shown in Section \ref{sec:simply}. However, our result is rewriting-based, and stems from structural rule elimination, not unification. While the typing they obtain is affine, ours is linear, by normalization of the exponential reduction. Moreover, while they need strong normalization of simply-typed terms to prove the correctness of their algorithm, we prove it independently, as a corollary of strong normalization. 

The Taylor expansion as defined by Ehrhard and Regnier~\cite{EhrhardRegnier08} can itself be seen as a form of linearization. It consists of assigning to every $ \lambda$-term, seen as an \emph{analytic function}, its Taylor expansion, i.e., the \emph{infinite} sum of its linear approximations, defined as resource terms. Only a finite number of such approximations (or one, in the rigid case) are correct. In this sense, what we do in this paper can be seen a way to build such an approximation \emph{from within} a calculus. A related approach is the one of \emph{polyadic approximations} \cite{mazza:pol}, whereas both linear and non-linear approximations are allowed. The type of approximation is closely related to the structural rules allowed. Our structural resource $\lambda$-calculus can be seen as a further refinement of that framework, where \emph{nested} structural rules are allowed.

While we presented our results without relying on abstract constructions, our work has been strongly influenced by the 2-dimensional categorical semantics of the $ \lambda$-calculus. In particular, the definition of structural resource $\lambda$-terms, the action of morphisms and their operational semantics come from a fine-grained analysis of the semantics of $ \lambda$-terms in the bicategory of \emph{generalized species of structures} \cite{fiore:esp, tao:gen, tao:prof, gal:prof, ol:intdist, ol:proof}.
\section{Conclusion}
In this paper, we introduced the structural resource $\lambda$-caluclus with its terminating and confluent rewriting system. Two reduction relations are present, namely the \emph{exponential} one, which performs linearization by structural rule elimination, and the \emph{linear} one, which instead performs linear $\beta$-reduction. We have shown how to embed qualitative type systems into our calculus and how to recover a classic result in the denotational semantics of linear logic, namely the \emph{extensional collapse}. 

We believe that the structural resource $\lambda$-calculus can be suitably extended to encode other qualitative systems of interest. As a first goal we will work on obtaining an encoding of
G\"odel's $\mathcal{T}$, then targeting dependent types and polymorphism. 

We have shown how the reduction semantics of some qualitative systems can be factorized into exponential and linear steps. Since the normalization of linear reduction holds by a simple combinatorial argument, the proof of normalization for those qualitative systems depends on the well-foundedness of exponential reduction, which we proved by a suitable reducibility argument. Would it be possible to obtain a more combinatorial proof? A positive answer to this question could also shed some new light on the so called G\"{o}del's Koan \cite{goedel:koan}, that is about finding natural measures that decrease under reduction of typed $ \lambda$-terms.  

From the point of view of denotational semantics, we plan to investigate the relationship of our linearization with a 2-dimensional version of the extensional collapse. We speculate that the normalization function $  s \to \mathsf{nf}^e(s)$ could be the syntactic presentation of a natural isomorphism expressing a \emph{proof-relevant} version of the collapse, arising from a 2-dimensional orthogonality construction \cite{gal:thesis, gal:proof}. 
The interactive flavour of our rewriting system also suggests a strong connection with \emph{game semantics}. In that context, it would be interesting to relate our construction to the collapse of concurrent games \cite{clair:collapse}. From this, we could then hope to exploit the connection between generalized species and concurrent games established in \cite{clair:esp} to formalize the relationship between the two collapse results.

\bibliographystyle{plain}
\bibliography{main}
\appendix

We sketch some of the proofs of the main results. Proofs are divided thematically.

\subsection*{Typing}

\begin{proposition}[Uniquness of Type Derivations]
Let $ \pi $ be a derivation of $\gamma \vdash s : \ty $ and $ \pi'$ of  $ \gamma \vdash s : \ty' $. Then $ \ty = \ty'$ and $  \pi = \pi'$.
\end{proposition}
\begin{proof}
By induction on $ s$. The only interesting cases are the application and list terms. We prove the application case, the list one being completely analogous. Let $ s = p \vec{q} $ and $  \pi = $
\[         \begin{prooftree}
\hypo{  \gamma_0 \vdash p : \tyl \multimap \ty} \hypo{\gamma_1 \vdash \vec{q} : \tyl}
\infer2{      \gamma_0 \otimes \gamma_1 \vdash p \vec{q} : \ty }
\end{prooftree}        \]
and $ \pi' =$
\[  \begin{prooftree}
\hypo{  \gamma'_0 \vdash p : \tyl' \multimap \ty} \hypo{\gamma'_1 \vdash \vec{q} : \tyl}
\infer2{      \gamma_0 \otimes \gamma_1 \vdash p \vec{q} : \ty }
\end{prooftree}     \]
with $\gamma_0 \otimes \gamma_1 = \gamma'_0 \otimes \gamma'_1 $. We want to prove that $ \gamma_{i} =\gamma'_{i} $. Let $  \gamma_i = x_1 : \tyl_1,\dots, x_n : \tyl_n $ and $ \gamma'_i = x_1 : \tyl'_1, \dots, x_n : \tyl'_n $. The context is relevant, meaning that $ \tyl_i \oplus \tyl'_i$ contains exactly the typing of all occurrences of $x_i $ in $ p \vec{q} $, ordered from the leftmost to the rightmost. Hence the same is true for $ \gamma_i, \gamma'_i$ and then we can conclude.  
\end{proof}

\begin{proposition}[Compositionality]\label{lem:comp}
 Let $ \gamma \vdash s : \ty $ and $  \theta : \delta \to \gamma, \theta' : \delta' \to \delta , f : \ty \to b, g : b \to c$. The following statements hold. 
 \begin{align*}
  &  (\gamma^{[\mu^f_s]})^{[\mu_{[f]s}^g]} \vdash [g]([f]s) : c &  &=&   &\gamma^{[\mu_{s}^{gf}]} \vdash [gf] s : c .&
  \\[0.3em]
   &(\delta'^{[\nu^{\theta}_s]})^{[\nu_{ \actr{s}{\theta} }^{\nu^{\theta}_s \theta'}]} \vdash \actr{\actr{s}{ \theta }}{\nu^{\theta}_s \theta'} : \ty& &=&  &\delta'^{[ \nu^{\theta \theta'}_{s} ] }\vdash \actr{s}{\theta \theta'} : \ty.&   
 \end{align*}
 \end{proposition}
 \begin{proof}
 We prove the interceding cases.
 \begin{enumerate}
 \item By induction on $ s $. If $ s = \la{\tysx^{\overline{h}}} p  $ with $ f = \overline{f} \multimap o $ and $ h = \overline{g} \multimap o $, we have that $ [g]([f] s) = \la{ \tysx^{ ((\overline{h}\overline{f}) \overline{g}}   } p $ and $ [gf] s = \la{\tysx^{ \overline{h} (\overline{fg})}} p $. We then conclude by associativity of morphism composition. 
 
 If $ s = \seqdots{s}{1}{k}  $ with $ f = \seq{  \alpha; f_1,\dots, f_l  } $ and $ g = \seq{\beta; g_{1},\dots, g_m } $ we have that $    [g] ([f] s ) = \seq{ [f_{\beta(1)} g_1]s_{\alpha (\beta (1))}    , \dots, [f_{\beta(m)} g_m]s_{\alpha(\beta (m)) }     }  $. We conclude again by definition of morphism composition and by applying the IH.
 \item By induction on $ s$, exploiting the former result. If $ s = x \overline{q}  $ then $ \theta = \seq{ \overline{f}  \multimap o } \otimes \zeta $ and $  \theta' = \seq{\overline{f'}\multimap o} \otimes \zeta'$. By IH we have that $  \overline{q} \{ \theta \} \{\nu^{\theta}_{\overline{q}} \theta'\}  =  \actr{\overline{q}}{\theta \theta'}$. By definition we have that 
 \[  \actr{\actr{s}{ \theta }}{\nu^{\theta}_s \theta'} = x [f'] ([f]\overline{q} \{ \theta \} \{\nu^{\theta}_{\overline{q}} \theta'\})     \]
 and
 \[\actr{s}{\theta \theta'} = [f'f] \actr{\overline{q}}{\theta \theta'}\]
 we then conclude by applying the IH and the former point of this lemma.
 \end{enumerate}
 \end{proof}

We also observe that covariant and contravariant action are interchangeable: $  [f](\ract{s}{\theta}) = \ract{([f]s)}{  \mu^{f}_s \theta  }  $.

\subsection*{Morphisms under reduction}

We study the behaviour of type morphism action under reduction.

Given a reduction step $ (\gamma \vdash s : \ty) \to_{\theta; f} ( \gamma' \vdash s' : \ty') $ and a morphism $ \theta' : \delta \to \gamma $, we would expect that we could infer some reduction step of the following shape: 
\[             (\delta^{[\nu^{\theta'}_s]}   \vdash \ract{s}{\theta'} : \ty )  \to_{ \eta, f  }  (\delta^{[ \nu^{\theta'\theta} ]} \vdash \ract{s'}{\theta \theta' ; \ty'}   )            \] 
for some ground morphism $ \eta $. However, it is easy to see that this is not the case, due to the contextual rule of reduction for $\lambda $-abstraction. We can still recover a convergence though: there exists $ t $ and appropriate morphisms $ \eta_1, \eta_2, i_1, i_2 $ s.t. \[ \ract{s}{\theta'} \to_{ \eta_1; i_1 } t \leftarrow_{\eta_2;i_2}  \ract{s}{\theta\theta'}  \] We formalize this result in the following lemma.


	\begin{lemma}[Actions under exp]\label{lem:actred} Let $ ( \gamma, x: \tyl \vdash s : \ty) \toexp_{\theta, g; f} ( \gamma', x : \tyl' \vdash s' : \ty') $ and $  g' : \vec{b} \to \tyl , h : \ty \to b$. 
	
	Then there exist morphisms
	\[   \mu_1 : \vec{b}\thinspace^{[\nu_{s,x}^{id, g'}\mu^{h}_{\ract{s}{id,g'},x}]} \to  \vec{c} \qquad \qquad      \mu_2 : \vec{b}\thinspace^{[\nu_{s',x}^{id, gg'}\mu^{hf}_{\ract{s'}{id, gg'},x}]} \to \vec{c}\]
	\[ \zeta_1 : \gamma\thinspace^{[\nu_s^{id, g'}\mu^{h}_{\ract{s}{id,g'}}]} \to \delta \qquad \qquad \zeta_2 : \gamma'\thinspace^{[\nu_{s'}^{id, gg'}\mu^{hf}_{\ract{s'}{id, gg'}}]} \to \delta \] and $  i : c \to b $ and a term $  \delta, x : \vec{c} \vdash t : c$ s.t. \[[h]\ract{s}{id, g'} \toexp_{  \zeta_1, \mu_1; if } t \quad  \text{ and } \quad [hf]\ract{s'}{id, gg'} \toexp_{\zeta_2, \mu_2; i} t  \] with \[ (\zeta_1, \mu_1) \nu_s^{id, g'}\mu^{h}_{\ract{s}{id,g'}}  =  (\zeta_2, \mu_2) \nu_{s'}^{id, gg'}\mu^{hf}_{\ract{s'}{id, gg'}} \theta  \]
 \end{lemma} 
 \begin{proof}
 By induction on the step $ s \toexp s'$. 
 
  Let  \[ \la{\tysx^{\tysf'}} s \toexp_{\nu^{\overline{f'}}, \nu^{\overline{f'}}_y; \nu^{\overline{f'}}_x \multimap o}  \la{x} \ract{s}{id,id, \overline{f'}}   \] with $ f = \nu^{\overline{f'}}_x \multimap o $ and $ \theta = \nu^{\overline{f'}} $. We have that $  h$ is of the shape $ \overline{h} \multimap o$, $ g =  \nu^{\overline{f'}}_y$ and \[ [h]\ract{s}{id, g'} = \la{ \tysx^{\overline{f'h}}  } \ract{s}{id, g', id}   \] Then \[  [hf]\ract{s'}{id, gg'} = \la{ \tysx^{\overline{\nu^{\overline{f'}}_x h}}  } \ract{s}{id, \nu'^{\overline{f'}}_y g', \overline{f'}}  \]
 
 By compositionality of the actions and by their interchange, we have the following reductions: 
\[\begin{tikzcd}
	&& {[hf]\ract{s'}{id, gg'}} \\
	\\
	{[h]\ract{s}{id, g'}} && {\lambda \overline{x} .s \{ id , g', \overline{f'h} \}} \\
	\\
	& {}
	\arrow["{\nu^{\nu^{\overline{f'}}\overline{h}}\nu^{g'} , (\nu^{\overline{f'}_x h}_x\nu^{g'}_{x}); \nu'^{\nu^{\overline{f'}}\overline{h} \nu_x^{g'}} \multimap o}", from=1-3, to=3-3]
	\arrow["{\nu^{\overline{f'h}} \nu^{g'}, \nu^{\overline{fh}}_x \nu^{g'}_y ; \nu^{\overline{f'h}}_x \nu^{g'}_x\multimap o}"', from=3-1, to=3-3]
\end{tikzcd}\] 
We set \[ t = {\lambda \overline{x} .s \{ id , g', \overline{f'h} \}} \]
\[  \zeta_1 =  \nu^{\overline{f'h}} \nu^{g'}  \qquad \zeta_2 =  \nu^{\nu^{\overline{f'}}\overline{h}}\nu^{g'}\]
\[              \mu_1 =             \nu^{\overline{fh}}_x \nu^{g'}_y \qquad \mu_2 =       \nu^{\overline{f'}_x h}_x\nu^{g'}_{x}    \]
We remark that in this case $ \mu^f_{s\{ id; g'\}}, \mu^{hf}_{\ract{s}{gg'}} $ are identities, since covariant action on lambda abstraction does not change the context. The equations are then satisfied by compositionality (Lemma \ref{lem:comp}).

Let $  s = y \overline{q} $ with $ y : \seq{ \overline{\ty} \multimap o} \vdash_v y : \overline{\ty } \multimap o $ and $ \delta, x : \tyl \vdash \overline{q} : \overline{\ty} $ Then $ s' = x \overline{q'} $ with $ \overline{q} \toexp_{\zeta; \overline{f}} \overline{q'} .$
 Since the term has an atomic type, we get $ h = f = id $. There are two possibilities.
 
  Let $  x = y$. Then $  g = \seq{ \overline{f} \multimap o } \oplus \zeta_{x} $,  $  g' = \seq{ \overline{h} \multimap o  } \oplus g'' $ for some morphism $ g''$ and $ \ract{s}{id, g'} = x [\overline{h}] \ract{\overline{q}}{g''} $. Also $ \ract{s'}{id, gg'}  = x [\overline{hf}] \ract{\overline{q'}}{ \zeta_x g'' }  $. By IH we get morphisms $ \eta_1, \eta_2, i $ and a sequence term $ \overline{t}$ s.t.
\[\begin{tikzcd}
	{ [\overline{h}]\overline{q} \{ id, g''\}} && { [hf]\overline{q'}\{ \zeta_x g'' \}} \\
	\\
	& \overline{t}
	\arrow["{\eta_1 ; \overline{if}}"', from=1-1, to=3-2]
	\arrow["{\eta_2, \overline{i}}", from=1-3, to=3-2]
\end{tikzcd}\]
 Then we can infer \[   x[\overline{h}]\overline{q} \{ id, g''\} \toexp_{ (\seq{\overline{if} \multimap o} \oplus \eta_1 );id} t = ( x\overline{t}) \] and \[ x[hf]\overline{q'}\{ \zeta_x g'' \} \toexp_{ (\seq{\overline{i} \multimap o} \oplus \eta_2 ); o} t = (x\overline{t}) \]
Then we set $ \zeta_1 , \mu_1 = \seq{\overline{if} \multimap o} \oplus \eta_1  $ and $ \zeta_2 , \mu_2 = \seq{\overline{f} \multimap o} \oplus \eta_2$.  The equations are satisfied by applying the IH and compositionality.

If $x \neq y $, we have again that $h = f = id  $. We have that $\ract{s}{id, g'} = y (\ract{\overline{q}}{id, g'}) $ and $ \ract{s'}{id, gg'} = y (\ract{\overline{q'}}{id, g'}).$ By IH we get morphisms $ \eta_1, \eta_2, i $ and a sequence term $ \overline{t}$ s.t.
\[\begin{tikzcd}
	{ \overline{q} \{ id, g'\}} && { \overline{q'}\{  id, gg' \}} \\
	\\
	& \overline{t}
	\arrow["{\eta_1 ; \overline{if}}"', from=1-1, to=3-2]
	\arrow["{\eta_2, \overline{i}}", from=1-3, to=3-2]
\end{tikzcd}\]

 Then we can infer that 
\[   y (\overline{q} \{ id, g'\})  \toexp_{\seq{\overline{if} \multimap o} \oplus \eta_1; o} t = ( y \overline{t})   \]
and
\[ y (\overline{q'} \{ id, gg'\})  \toexp_{\seq{\overline{i} \multimap o } \oplus \eta_2; o} t = y \overline{t} \]
Then we set $ \zeta_1 , \mu_1 = \seq{\overline{if} \multimap o} \oplus \eta_1  $ and $ \zeta_2 , \mu_2 = \seq{\overline{f} \multimap o} \oplus \eta_2$.  The equations are satisfied by applying the IH.

If $  s = \la{\overline{y}^{\overline{l}}} p  $ and $  s' =  \la{\overline{y}^{\overline{l'l}}} p $ with $ p \to_{\theta, g, l'; o} p' $, we have that $  f = id \multimap id , h = \overline{h} \multimap o $.

By definition of actions, 
\[       [h]\ract{s}{id, g} = \la{\overline{y}^{ \nu^{id,id,g'}_y l h }} \ract{p}{id,id,g'} \qquad [hf]\ract{s'}{id,gg'}   =      \la{\overline{y}^{ \nu^{id,id, gg'}_y l'l h }} \ract{p'}{id,id,gg'}             \]
 By IH there exists $ t, \zeta, \zeta', \mu, \mu'$ and a term $ t $ s.t.  
\[            \ract{p}{id, id, g'} \toexp_{\zeta, \mu, \zeta_{y}; o} t \quad \text{ and } \quad  \ract{p'}{id,id,gg'; o} \toexp_{\zeta', \mu', \zeta'_y; o} t                             \] then, by contextuality
\[             \la{\overline{y}^{\nu^{id,id,g'} l h}} \ract{p}{id, id, g'} \toexp_{\zeta, \mu ; o} \la{ \overline{y}^{\nu^{id,id, ,g'}_y \zeta_y l h} }  t \] \text{ and } \[ \quad \la{\overline{y}^{\nu^{id,id, gg'}_y l'l h}} \ract{p'}{id,id,gg'} \toexp_{\zeta', \mu', \zeta'_y} \la{\overline{y}^{\nu^{id,id, gg'}_y\zeta'_y l'l h } } t                                        \]
By the IH we know that 
\[ \nu^{id,id, ,g'}_y \zeta_y =  \overline{l'}\nu^{id,id, gg'}_y\zeta'_y l'  \]
we can then conclude that $\la{ \overline{y}^{\nu^{id,id, ,g'}_y \zeta_y l h} }  t = \la{\overline{y}^{\nu^{id,id, gg'}_y\zeta'_y l'l h } } t                                       . $

If $ s = (\la{\overline{y}^{\overline{f_1}}} p ) \overline{q} $ then $ f = h = id  $. 
We have two possible cases.

Let $   s \toexp_{\theta, g; id} s'  $ with $ s' = (\la{\overline{y}^{\overline{f_2f_1}}} p' ) \overline{q} $ and $  p \toexp_{\theta'; g'', f_2} p'  $. Then the result is a direct consequence of the IH.

Let $   s \toexp_{\theta, g; id} s'  $ with $ s' = (\la{\overline{y}^{\overline{f_2f_1}}} p' ) \overline{q'} $ with $ \overline{q} \toexp_{\theta' g, f_2} \overline{q'} $. 

By definition of contravariant action, we get 
\[      \ract{s}{id, g'} =   \ract{(\la{\overline{y}^{\overline{f_1}}} p )}{id, g'_1} \ract{\overline{q}}{id, g'_2}  \]
and
\[      \ract{s'}{id, gg'}  =    \ract{(\la{\overline{y}^{\overline{f_2f_1}}} p )}{id, g'_1} \ract{\overline{q'}}{id, gg'_2}  \]

By IH there exist a sequence term $ \overline{t}$ and morphisms $ \eta_1,\eta_2, \rho_1,\rho_2 , \overline{i'}$ s.t. 
\[    \ract{\overline{q}}{id, g'_2} \toexp_{\eta_1, \rho_1; \overline{i'f_2}} \overline{t}               \] 
and
\[  \ract{\overline{q'}}{id, g g'_2} \toexp_{\eta_2, \rho_2; \overline{i'}} \overline{t}              \]
hence we can conclude that 
\[         \ract{s}{id,g'}  \toexp_{id \otimes (\eta_1,\rho_1) ; o}   (\la{\overline{y}^{\overline{i'f_2f_1}}} p ) \overline{t}                 \]
and 
\[         \ract{s}{id,gg'}  \toexp_{id \otimes (eta_2,\rho_2) ; o}   (\la{\overline{y}^{\overline{i'f_2f_1}}} p ) \overline{t}                \]
we then set $ t =  (\la{\overline{y}^{\overline{i'f_2f_1}}} p ) \overline{t}      , \zeta_1 = id \otimes \eta_1, \zeta_2 = id \otimes \zeta_2, \mu_1 = id \oplus \rho_1, \mu_2 = id \oplus \rho_2 , i = id $. The coherence condition is satisfied as a direct consequence of the IH.

The list and sequence term cases are direct consequences of the IH.
 \end{proof}


\begin{lemma}[Associativity of Substitution]\label{lem:subassoc} let $ \gamma, x : \tyl_1, y : \vec{b} \vdash s : \ty  $, $\delta_1, y : \tyl_2 \vdash \vec{t} : \vec{b} $ and $ \delta_2 \vdash \vec{q} : \tyl_1 \oplus \tyl_2 $. We have that \[ \subst{\subst{s}{x}{\vec{t}}}{y}{\vec{q}^{[\sigma_{s, \vec{t}}]}} = \subst{\subst{s}{y}{\vec{q}_1} }{x}{   (\subst{ \vec{t}}{y}{\vec{q}_2})^{[\sigma_{s,\vec{q}_1}]}  }    \]
with \[ \sigma_{\subst{s}{x}{\vec{t}}, \vec{q}} (\sigma_{s,\vec{t}} \otimes id ) = \sigma_{\subst{s}{y}{\vec{q}_1},\subst{ \vec{t}}{y}{\vec{q}_2}^{[\sigma_{s,\vec{q}_1}]} } ( \sigma_{s,\vec{q}_1} \otimes \sigma_{\vec{t}, \vec{q}_2})  \]
where $ \vec{q} = \vec{q}_1 \oplus \vec{q}_2 $, the decomposition being the one induced by the typing (Figure \ref{fig:linsub}).
\end{lemma}
 \begin{proof}
 By induction on the structure of $s $. 
 
 If $ s = x $ then $ \vec{t} = \seq{t} \cdot \vec{t}' $ for some term $ t$ and list term $  \vec{t}$, with $ \subst{s}{x}{\vec{t}} = t  $. Hence
 \[ \subst{\subst{s}{x}{\vec{t}}}{y}{\vec{q}^{[\sigma_{s, \vec{t}}]}} = \subst{t}{y}{\vec{q}}\]
 and
 \[      \subst{\subst{s}{y}{\vec{q}_1} }{x}{   (\subst{ \vec{t}}{y}{\vec{q}_2})^{[\sigma_{s,\vec{q}_1}]}  } =  \subst{t}{y}{\vec{q}_2} = \subst{t}{y}{\vec{q}} \]
 since $ s$ does not contain any occurrence of the variable $y $.
 
 If $  s = z $ with $  y \neq x \neq y$ then $ \vec{q_2} = \vec{q} , \vec{q}_1 = \seq{} $ and \[ \subst{\subst{s}{x}{\vec{t}}}{y}{\vec{q}} = z = \subst{\subst{s}{y}{\vec{q}_1} }{x}{   (\subst{ \vec{t}}{y}{\vec{q}_2})^{[\sigma_{s,\vec{q}_1}]}  } .\]
 
 If $ s = \la{y^f}  p$ then 
 \[ \subst{\subst{s}{x}{\vec{t}}}{y}{\vec{q}^{[\sigma_{s, \vec{t}}]}} = \la{y^{ \sigma_{\subst{p}{x}{\vec{t}}, \vec{q}} (\sigma_{s,\vec{t}} \otimes id )}}\subst{\subst{p}{x}{\vec{t}}}{y}{\vec{q}^{[\sigma_{p, \vec{t}}]}} \]
 and
 \[  \subst{\subst{s}{y}{\vec{q}_1} }{x}{   (\subst{ \vec{t}}{y}{\vec{q}_2})^{[\sigma_{s,\vec{q}_1}]}  }  = \la{y^{\sigma_{\subst{p}{y}{\vec{q}_1},\subst{ \vec{t}}{y}{\vec{q}_2}^{[\sigma_{p,\vec{q}_1}]} } ( \sigma_{p,\vec{q}_1} \otimes \sigma_{\vec{t}, \vec{q}_2})}} \subst{\subst{p}{y}{\vec{q}_1} }{x}{   (\subst{ \vec{t}}{y}{\vec{q}_2})^{[\sigma_{p,\vec{q}_1}]}  }  . \]
We can then conclude by applying the IH.

If $  s = p \vec{u}$ then we assume the following typing
\[  \gamma_1, x : \vec{a}_1, y : \vec{b}_{1,1} \vdash p : \vec{c} \multimap b 
\qquad     \gamma_2, x : \tyl_2 , y : \vec{b}_{1,2} \vdash \vec{u} : \vec{c}   \]
\[    (\delta_1 \otimes \delta_2)  , y : \vec{b}_{2,1} \oplus \vec{b}_{2,2}   \vdash \vec{t}_1 \oplus \vec{t}_2 : \tyl_1 \oplus \tyl_2                            \]
\[  (\delta'_1 \otimes \delta'_2 \delta'_3 \otimes \delta'_4) \vdash \vec{q}_{1,1} \oplus \vec{q}_{1,2} \oplus \vec{q}_{2,1} \oplus \vec{q}_{2,2} : \vec{b}_{1,1} \oplus \vec{b}_{1,2} \oplus \vec{b}_{2,1} \oplus \vec{b}_{2,2} \]
with $ \vec{t}_1 \oplus \vec{t}_2 = \vec{t} $ and $ \vec{q}  = \vec{q}_{1,1} \oplus \vec{q}_{1,2} \oplus \vec{q}_{2,1} \oplus \vec{q}_{2,2}$. The decomposition of the list terms follow the decomposition of contexts induced by the application rule.

By definition
\[\subst{\subst{s}{x}{\vec{t}}}{y}{\vec{q}^{[\sigma_{s, \vec{t}}]}} = \subst{\subst{p}{x}{\vec{t}_1}}{y}{\vec{q}_1^{[\sigma_{p, \vec{t}_1}]}} \subst{\subst{\vec{u}}{x}{\vec{t}_2}}{y}{\vec{q}_2^{[\sigma_{\vec{u}, \vec{t}_2}]}}\]
and
\[  \subst{\subst{s}{y}{\vec{q}_1} }{x}{   (\subst{ \vec{t}}{y}{\vec{q}_2})^{[\sigma_{s,\vec{q}_1}]}  } =  \subst{\subst{p}{y}{\vec{q}_{1,1}} }{x}{   (\subst{ \vec{t}_1}{y}{\vec{q}_{1,2}})^{[\sigma_{p,\vec{q}_{1,1}}]}  } \subst{\subst{\vec{u}}{y}{\vec{q}_{2,1}} }{x}{   (\subst{ \vec{t}_2}{y}{\vec{q}_{2,2}})^{[\sigma_{\vec{u},\vec{q}'_1}]}  }\]
with $\vec{q}_i = \vec{q}_{i,1} \oplus \vec{q}_{i,2} .$ By the IH, we have that
\[    \subst{\subst{p}{x}{\vec{t}_1}}{y}{\vec{q}_1^{[\sigma_{p, \vec{t}_1}]}} =  \subst{\subst{p}{y}{\vec{q}_{1,1}} }{x}{   (\subst{ \vec{t}_1}{y}{\vec{q}_{1,2}})^{[\sigma_{p,\vec{q}_{1,1}}]}  }    \]
and
\[ \subst{\subst{\vec{u}}{x}{\vec{t}_2}}{y}{\vec{q}_2^{[\sigma_{\vec{u}, \vec{t}_2}]}} = \subst{\subst{\vec{u}}{y}{\vec{q}_{2,1}} }{x}{   (\subst{ \vec{t}_2}{y}{\vec{q}_{2,2}})^{[\sigma_{\vec{u},\vec{q}'_1}]}  }  \]
with \[ \sigma_{\subst{p}{x}{\vec{t}_1}, \vec{q}_1} (\sigma_{p,\vec{t}_1} \otimes id ) = \sigma_{\subst{p}{y}{\vec{q}_{1,1}},\subst{ \vec{t}_1}{y}{\vec{q}_{1,2}}^{[\sigma_{p,\vec{q}_{1,1}}]} } ( \sigma_{s,\vec{q}_{1,1}} \otimes id)  \]
and
\[ \sigma_{\subst{\vec{u}}{x}{\vec{t}_2}, \vec{q}_2} (\sigma_{\vec{u},\vec{t}_2} \otimes id ) = \sigma_{\subst{\vec{u}}{y}{\vec{q}_{1,2}},\subst{ \vec{t}_2}{y}{\vec{q}_{2,2}}^{[\sigma_{\vec{u},\vec{q}_{2,1}}]} } ( \sigma_{\vec{u},\vec{q}_{2,1}} \otimes id) . \]

We recall that the permutation induced by the application is defined exploiting the natural transformation:
\[  \tau_{ \gamma_1, \gamma_2, \delta_1, \delta_2 } : (\gamma_1 \otimes \gamma_2) \otimes (\delta_1 \otimes \delta_2) \to (\gamma_1 \otimes \delta_1) \otimes (\gamma_2 \otimes \delta_2)    \]

In particular, we have that 
\[ \sigma_{\subst{s}{x}{\vec{t}}, \vec{q}} =  \tau_{(\gamma_1 \otimes \delta_1) , (\gamma_2, \otimes \delta_2), \delta'_1, \delta'_2 } ( \sigma_{\subst{p}{x}{\vec{t}_1}, \vec{q}_1} \otimes \sigma_{\subst{\vec{u}}{x}{\vec{t}_2}, \vec{q}_2}  )   \qquad \sigma_{s,\vec{t}} =  \tau_{ \gamma_1,\gamma_2, \delta_1, \delta_2 } ( \sigma_{p,\vec{t}_1} \otimes \sigma_{\vec{u}, \vec{t}_2} ) \]
and 
\[ \sigma_{\subst{s}{y}{\vec{q}_1},\subst{ \vec{t}}{y}{\vec{q}_2}^{[\sigma_{s,\vec{q}_1}]} } =  \] \[ \tau_{ (\gamma_1 \otimes \delta'_{1,1}),  (\gamma_2 \otimes \delta'_{1,2}) , (\delta_1 \otimes \delta'_{2,1})^{[\sigma_{p,\vec{q}_{1,2}}]}, (\delta_1 \otimes \delta'_{2,2})^{[\sigma_{\vec{u},\vec{q}_{2,2}}]} } (\sigma_{\subst{p}{y}{\vec{q}_{1,1}},\subst{ \vec{t}_1}{y}{\vec{q}_{1,2}}^{[\sigma_{p,\vec{q}_{1,2}}]} } \otimes \sigma_{\subst{\vec{u}}{y}{\vec{q}_{2,1}},\subst{ \vec{t}_2}{y}{\vec{q}_{2,2}}^{[\sigma_{\vec{u},\vec{q}_{2,2}}]} } ) \]
\[ \]
\[ \sigma_{s,\vec{q}_1} = \tau_{ \gamma_1, \gamma_2,  \delta'_{1,1}, \delta'_{1,2} } ( \sigma_{p,\vec{q}_{1,1}} \otimes \sigma_{\vec{u},\vec{q}_{2,1}} ).\]
The result is then a consequence of the neutrality of $\tau $ and the monoidality of the tensor product of typing contexts, applying the IH.
 \end{proof}
 
 \begin{lemma}\label{lem:acto}
 Let $  \gamma, x : \tyl \vdash s : \ty $ and $ \delta \vdash \vec{t} : \tyl $ with $\theta : \gamma' \to \gamma $ and $ \theta' : \delta' \to \delta  $. we have that
 \[            \subst{s \{ \theta, id\}}{x}{[\nu^{\theta}_x] \ract{\vec{t}}{\theta'}} = \subst{s}{\vec{x}}{\vec{t}} \{\theta \otimes \theta' \}              \]
 \end{lemma}
 \begin{proof}
 By induction on the structure of $s  $. We prove the interesting cases.
 
 If $ s = x $ then $\theta  $ is a sequence of terminal morphisms and  $ \ract{s}{\theta, id} = x$. Then the result is immediate since $\subst{s \{ \theta, id\}}{x}{[\nu^{\theta}_x]\vec{t}} = t $ where we assume that $ \vec{t} = \seq{t} \cdot \vec{t}' $ for some term $ t$ and bag $\vec{t}' $.
 
 If $  s = \la{\overline{y}^{\tysf}} s'  $ we have that 
 \[     \subst{s \{ \theta, id\}}{x}{[\nu^{\theta}_x] \ract{\vec{t}}{\theta'}} =  \la{y^{\nu^{\theta}_y}} \subst{s' \{ \theta,id, id\}}{x}{[\nu^{\theta}_x] \ract{\vec{t}}{\theta',id}}\]
 and
 \[\subst{s}{\vec{x}}{\vec{t}} \{\theta \otimes \theta' \}  = \la{y^{\nu^{(\theta \otimes \theta')}_y}} \subst{s'}{\vec{x}}{\vec{t}} \{(\theta,id) \otimes (\theta',id) \}  \]
 By the IH we have that
 \[ \subst{s' \{ \theta,id, id\}}{x}{[\nu^{\theta}_x] \ract{\vec{t}}{\theta',id}} = \subst{s'}{\vec{x}}{\vec{t}} \{(\theta,id) \otimes (\theta',id) \} \]
 we can then conclude.
 
 If $ s = p \vec{q} $ then  $ \gamma = \gamma_1 \otimes \gamma_2 $ and $ \tyl = \tyl_1 \oplus \tyl_2 $ with $ \gamma_1, x : \tyl_1 \vdash p : \vec{b} \multimap \ty $ and $ \gamma_2, x : \tyl_2 \vdash \vec{q} : \vec{b}  $ , for some contexts $\gamma_1, \gamma_2 $, intersection types $  \tyl_1, \tyl_2, \vec{b}$. By definition we have that 
 \[   \subst{s \{ \theta, id\}}{x}{[\nu^{\theta}_x] \ract{\vec{t}}{\theta'}} = \] \[ \subst{\ract{p}{\theta_1, id}}{x}{  [\nu^{\theta_1}_x]\ract{\vec{t}_1}{\theta'_1}  } \subst{\ract{\vec{q}}{\theta_2}}{x}{  [\nu^{\theta_2}_x]\ract{\vec{t}_}{\theta'_2}  }     \]
 where $ \theta = \theta_1 \otimes \theta_2 $ and $ \theta' = \theta'_2 \otimes \theta'_2 $, the decomposition being the unique one induced by the typing. We then conclude by applying the IH.
  \end{proof}
 
 \begin{lemma}\label{lem:morpsubs}
 Let $ \gamma, x : \tyl \vdash s : \ty $ and $  \delta\vdash_q \vec{q} : \tyl $. The following statements hold.
 \begin{enumerate}
 \item If $ s \to_{\theta, g;f} s'   $ then $ \subst{s}{x}{\vec{q}} \to^\ast_{ \theta';f} u \cong  \subst{\ract{s'}{id; g}}{x}{ [\nu'] \vec{q} } $ s.t. $ \sigma_{\ract{s'}{id; g}, [\nu']\vec{q}} (\nu \theta \otimes \nu'^{\star})  =  \theta' \sigma_{s, \vec{q}}     $.
 \item $ if \vec{q} \to_{\theta; \seq{id;\vec{f}}} \vec{q'} $ then $ \subst{s}{x}{\vec{q}} \to^\ast_{ \theta' ;f} u \cong \subst{ \ract{s}{\seq{id;\vec{f}}}} {x}{\vec{q}'} $ s.t. $  \sigma_{\ract{s}{\seq{id;\vec{f}}} \vec{q}'} (\nu \otimes \theta) = \theta' \sigma_{s, \vec{q}}   $.
 \end{enumerate}
 \end{lemma}
\begin{proof}
\begin{enumerate}
\item By induction on $ s \to_{\theta, g;f} s' $. Let $  s = (\la{\overline{x}}p) \overline{q}  $ and $ s' = \subst{p}{\tysx}{\overline{q}} $. We conclude by applying Lemma   \ref{lem:subassoc}. 

Let $ s = \la{\overline{y}^{\overline{f}}} p $ and $s' = \la{\overline{y}} \ract{p}{id, \overline{f}} $. Then $ \subst{s}{x}{\vec{t}}  =   \la{ \overline{y}^{\sigma \overline{f}}} \subst{p}{x}{\vec{t}} $ and $ \subst{\ract{s'}{id; g}}{x}{ [\nu'] \vec{q} } = \la{\overline{y}^{\sigma'}} \subst{\ract{p}{id, \overline{f}}}{x}{[\nu_x^{\overline{f}}]\vec{t}}  $. we then have that $  \subst{s}{x}{\vec{t}} \to \ract{\la{\overline{y}} \subst{p}{x}{\vec{t}}}{  id, \sigma' \overline{f}  } $. We then conclude by Lemma \ref{lem:acto}.
\item By induction on the structure of $ s$.
\end{enumerate}
\end{proof}

\begin{remark}\label{rem:linsubj}
We observe that given a linear step $ s \toli_{\theta;f} s' $ the morphism $ \theta $ is a ground permutation and $f = id $. Hence, linear reduction almost satisfy subject reduction, up to a permutation of the typing context.
\end{remark}

\begin{lemma}\label{lem:actredlin}
Let $ (\gamma \vdash s : \ty ) \toli_{\sigma; id}(\gamma' \vdash s' : \ty) $ and $ \theta' : \delta \to \gamma $, $ f' : \ty \to b $. The following statements hold.
\begin{enumerate}
\item There exists a morphism $\zeta $ s.t. $   (\gamma^{[\mu^f_s]} \vdash [f]s : b ) \toli_{\zeta; id}(\gamma'^{\mu^{f}_{s'}} \vdash [f]s' : b)  $ and $ \zeta \mu^f_s= \mu^f_{s'} \theta $.
\item There exists a morphism $\zeta $ s.t.  $                    (\delta^{[\nu^{\theta}_s]} \vdash \ract{s}{\theta} : \ty ) \toli_{\zeta; id}(\delta^{[\nu^{\sigma\theta}_{s'}]} \vdash \ract{s'}{\sigma \theta} : \ty)          $ and $\zeta \nu^{\theta}_s = \nu^{\sigma\theta}_{s'} $.
\end{enumerate}
\end{lemma}

	\subsection*{Reducibility}

\begin{lemma}
$\mathsf{SN} $ is saturated.
\end{lemma}
\begin{proof}
We have to prove that given $  \ract{s}{id; \overline{fg}} \in \mathsf{SN}   $ with $ \overline{q} \in \mathsf{SN} $ and $ \overline{q} \to_{\theta; \overline{f}} \overline{q'} $ then $ (\la{\tysx^{\tysg}} s) \overline{q} \in \mathsf{SN} $. First, we observe that if $ \ract{s}{id; \overline{fg}} \in \mathsf{SN}  $, then, in particular $ s \in \mathsf{SN} $. Otherwise, by Lemma \ref{lem:actred}, from an infinite chain for $s $ we could build one for $ \ract{s}{id; \overline{fg}} $. From this we can infer that $ \la{ \tysx^{\tysg} } s \in \mathsf{SN} $. We observe that the only other possible reductions for $ (\la{\tysx^{\tysg}} s) \overline{q}$ must come from $\overline{q} $. By hypothesis, we can then conclude.
\end{proof}
\begin{lemma}\label{lem:redusn}
We have that $  \mathsf{Red}(\ty) \subseteq \mathsf{SN} .  $
\end{lemma}
\begin{proof}
By induction on the structure of $\ty $. If $ \ty = o$ then $\mathsf{Red}(\ty) = \mathsf{SN} $

If $ \ty = \tysa \multimap o $ then \[ \mathsf{Red}(\ty) = \{   s \mid \text{ for all } \overline{q} \in \mathsf{Red}(\tysa), s \overline{q} \in \mathsf{Red}(o)   \}  \]
Hence, since $ s \overline{q} \in \mathsf{Red}(o) = \mathsf{SN} $, in particular $ s \in \mathsf{SN} $. The other cases are direct corollary of the IH.
\end{proof}

\begin{lemma} 
 Let $ \gamma \vdash s : \ty $ and $ \theta : \delta \to \gamma, f : \ty \to b . $ The following statements hold. 
 \begin{enumerate}
  \item $  [f]s \in \mathsf{Red}(b) . $
 \item $    \ract{s}{\theta} \in \mathsf{Red}(\ty) .   $
 \end{enumerate} 
\end{lemma}
\begin{proof}
We give the proof for the interesting cases.
\begin{enumerate}
\item By induction on $ s$. If $  s  = x \overline{q} $, with  $  x : \seq{\tysa \multimap o} \vdash_v x : \tysa \multimap o $ and $  \delta \vdash \overline{q} : \overline{\ty} $ then $ f = id_o $ and we shall prove that $ x \overline{q} \in \mathsf{Red}(o) $. By IH we have that $  [\tysf]\overline{q} \in \mathsf{Red}(\overline{b}) $ for any $ \tysf : \tysa \to \overline{b} $. Hence, in particular, $  [\tysa] \overline{q} \in \mathsf{Red}(\overline{\ty})  $ and then by Lemma\ref{lem:redusn}, $\overline{q} \in \mathsf{SN}  $. Then we can conclude, since if $\overline{q}  $ is strongly normalizing, then also $ x \overline{q}$ is. 

Let $ s = \la{\tysx^{\tysg}} p $, with $ \tysg : \overline{a} \to \overline{c}  $ and $ \gamma, \tysx : \overline{c} \vdash p : o $. Then $  f  = \tysf \multimap o $ for some sequence of list morphisms $ \overline{f} $ and we need to prove that $  \la{\tysx^{ \overline{gf}  }} p \in \mathsf{Red}( \overline{b} \multimap o ) $, \textit{i.e.}, for any $  \overline{q} \in \mathsf{Red}(\overline{b}) $, we have that $ (\la{\tysx^{ \overline{gf}  }} p) \overline{q} \in \mathsf{Red}(o) $. By hypothesis we have that $ \ract{s}{id; \theta} \in \mathsf{Red}(o) $ for any morphism $\theta $ and by Lemma \ref{lem:redusn} we have that $  \overline{q} \in \mathsf{SN} $. Hence, by saturation, we can conclude.

Let $ s  = ( \la{\tysx^{\tysg}} p ) \overline{q} $ with $ \gamma, x : \overline{b} \vdash p : o $, $ \delta \vdash \overline{q} : \overline{a}$ and $ \tysg : \overline{\ty} \to \overline{b}$. Then $ f  = id_o $ and we need to prove that $ ( \la{\tysx^{\tysf}} p ) \overline{q}  \in \mathsf{SN}$. By IH we have that $ [\overline{\ty} \multimap o] \la{\tysx^{\tysg}} p = \la{\tysx^{\tysg}} p \in \mathsf{Red}(\overline{\ty} \multimap o) $ and $ [\overline{\ty}] \overline{q} = \overline{q} \in \mathsf{Red}(\overline{\ty}) $. Then, by definition, $s \in \mathsf{Red}(o) $.
\item  By induction on $ s$, exploiting the former point of this lemma. If $     s  = x \overline{q}  $ with  $  x : \seq{\tysa \multimap o} \vdash_v x : \tysa \multimap o $ and $  \delta \vdash \overline{q} : \overline{\ty} $ we have that $ \theta = \seq{},\dots, \seq{\overline{f} \multimap o}, \dots, \seq{} $ with $\tysf : \overline{\ty} \to \overline{b} $. By hypothesis we have that $  [\overline{f}] \overline{q} \in \mathsf{Red}(\overline{b})  $. By Lemma \ref{lem:redusn}, we have that  $ [\overline{f}] \overline{q}  \in \mathsf{SN}$. Then we can conclude, since if $\overline{q}  $ is strongly normalizing, then also $ x \overline{q}$ is. 
\end{enumerate}
\end{proof}

\subsection*{Isomorphism and Commutation}

\begin{lemma}\label{lem:isomone}
Let $ \gamma \vdash s : \ty $ and $  \theta : \delta \to \gamma$, $f : \ty \to b $. The following statements hold.
\begin{enumerate}
 \item If $  s\cong s'$ then $ [f]s \cong [f] s' $. 
\item If $s \cong_{\sigma} s'$ then $  \ract{s}{\theta} \cong \ract{s'}{\sigma\theta} $.
\end{enumerate}
\end{lemma}
\begin{proof}
By induction on $s $. 
\end{proof}

\begin{lemma}
Let $ s\cong s' $ and $   s \toexp p   $ (resp. $ s \toli p$). Then there exists a term $ p' $ s.t. $ s' \toexp p'  $ (resp. $s' \toli p' $) and $ p \cong p' $.
\end{lemma}
\begin{proof}
By induction on $s \toexp p $ (resp. $ s \toli t $). We prove the interesting cases. Let $  s = \la{\tysx^{\tysf}} t  $ and $  p = \la{\tysx} \ract{t}{id , \tysf} $. We have that $  s \cong s'$ then $ s' = \la{\tysx^{\sigma \tysf}} t' $ with $ t \cong t' $ and $ \sigma$ permutation. Then we apply Lemma \ref{lem:isomone} and conclude. Let $ s = ( \la{\tysx^{\tysf}}t) \overline{q} $ and $ p = (\la{\tysx^{\overline{fg}}}t) \overline{q'}  $ with $ \overline{q} \toexp_{\theta; \overline{g}} \overline{q'} $. Since $ s \cong s' $ then either $   s ' = (\la{\tysx^{\sigma\tysf}}t') \overline{q} $ with $ t \cong t' $ or $ s' = (\la{\tysx^{\tysf}}t) \overline{q''} $ with $ \overline{q} \cong \overline{q''} $. In both cases we conclude by applying the IH and then contextuality of the reduction.
 \end{proof}

\begin{lemma}
Let $ \gamma, x : \tyl \vdash s : b $ and $ \delta\vdash \vec{t} : \tyl $ with $   \subst{s}{x}{\vec{t}} \toexp u    $. Then there exist $ p, \vec{q} $ and morphism $  g $ s.t. $ (\la{x} s) \vec{t}  \toexp^\ast   (\la{x^g}  p) \vec{q}  $ and  $ \subst{\ract{p}{g}}{x}{[\nu^g]\vec{q}} \cong u $.
 \end{lemma}
\begin{proof} We prove the interesting cases. By induction on $s $.

 If $ s = x \overline{v} $, we have that $ \vec{t} = \seq{t} \cdot \vec{t}' $ and $ \subst{s}{x}{\vec{t}} = t \subst{\overline{v}}{x}{\vec{t}'} $. We reason by cases on the step of exponential reduction. Let $ t \toexp_{\theta; f \multimap o} t' $. Then $\subst{s}{x}{\vec{t}} \toexp t [f]\subst{\overline{v}}{x}{\vec{t}'} $. Then we set $ \vec{q} = \seq{t'} \cdot \vec{t}' $, $ p = x$ and $ g = f $. We can conclude since
\[        (\la{x } x \overline{v}) \seq{t} \cdot \vec{t}' \toexp (\la{x^f} x \overline{v}) \seq{t'} \cdot \vec{t}'          .     \]
Otherwise, we have that $ \subst{\overline{q}}{x}{\vec{t}'} \toexp u $. In that case, we apply the IH and conclude by contextuality.

Let $ s = \la{\tysy^{\tysf}} s' $, with $ \tysy \cap \tysx = \emptyset   $. Then $  \subst{s}{x}{\vec{t}}  =  \la{ \tysy^{\sigma\tysf}} \subst{s'}{x}{\vec{t}}  $. There are two possible cases. Let \[\subst{s}{x}{\vec{t}} \toexp  \la{ \tysy} \ract{\subst{s'}{x}{\vec{t}}}{id; \sigma \tysf}\]
Then we can conclude, since $  s \toexp^\ast  \la{\tysy} \ract{s'}{id; \tysf} $ and by Lemma \ref{lem:isomone} we have that $\ract{s'}{id; \tysf} \cong \ract{s'}{id; \sigma\tysf}  $. Otherwise, $           \subst{s}{x}{\vec{t}} \toexp \la{\tysy^{\tysg\tysf}} s''  $ with $ \subst{s'}{x}{\vec{t}} \toexp_{\theta', \tysg ; o}  s''  $. BY IH we have a morphism $  \tysh $ and terms $ p, \vec{q} $ s.t. $    s'' \cong \subst{\ract{p}{id;h}}{x}{[\nu^h] \vec{q}}        $ and $    (\la{x} s')\vec{t} \toexp^\ast (\la{x^h} p)\vec{q}    $. We can then conclude by contextuality.
\end{proof}

\begin{theorem}[Commutation]\label{theo:commutationap}
Let $  s \toli^{\ast} t \toexp^\ast t'  $ Then there exist terms $  u, u' $ s.t. $ s \toexp^\ast u \toli^\ast u'  $ with $ t' \cong u'  $. 
\end{theorem}
\begin{proof}
By induction on $s \toli^{\ast} t $. We prove the interesting cases.

If $ s = (\la{\tysx } p) \overline{q}$ and $ t = \subst{p}{\overline{x}}{\overline{q}} $, we apply the former lemma and conclude. 

If $ s= \la{\tysx^{\tysf}} p $ then $ t = \la{\tysx^{ \sigma \tysf }} p'$ with $  p \to_{\theta; \sigma; o} p' $ for some permutation $ \theta$ and $ \sigma$. There are two possible cases. Let $ t' = \la{\tysx} \ract{p'}{id; \sigma \tysf}  $. Then $  s \toexp  \la{\tysx} \ract{p}{id; f} $ and we conclude by Lemma \ref{lem:actredlin} since $ \ract{p}{id; f}  \toli \ract{p'}{id; \sigma f} $. Hence $ u' = \la{x^\sigma} \ract{p'}{id; \sigma f} $. 
\end{proof}

\subsection*{Confluence}

\begin{theorem}
The structural reduction is locally confluent.
\end{theorem}
\begin{proof}
Given $ s \to^{\theta_1; f_1} t_1 $ and $ s \to^{\theta_2; f_2} t_2 $ we have to prove that there exists a term $  t $ and morphisms $  \theta'_1, \theta'_2, f'_1, f'_2 $ s.t. $     t_1 \to^{\theta'_1; f'_1} t    $ and $   t_2 \to^{\theta'_2; f'_2} t  $, with $ \theta'_2  \theta_2 = \theta'_1 \theta_1   $ and $ f_1 f'_1 = f_2f'_2 $.

We proceed by induction on $ s \to^{\theta_1;f_1} t_1 $. Let $  s = \la{\overline{x}^{\overline{f}}} p $ and $  t_1 = \la{\overline{x}}  \ract{p}{id; \overline{f}}$ with $ \theta_1 = \nu^{id,\overline{f}} , f_1 = \nu_{x}^{id,\overline{f}}$ and $s \to_{\nu^{id, \overline{f}}; \nu^{id,\overline{f}}_x \multimap o} t_1  $. Then we have that $ t_2 = \la{\overline{x}^{\overline{gf}}} p' $ with $ p \to^{\theta_2, g; o} p' $ and then $f_2 = id $. We perform the step $   t_2 \to^{\nu^{id, \overline{gf}}; \nu^{id, \overline{gf}}_x \multimap o   } \la{x} \ract{p'}{id;\overline{gf}}      $. Now we need to close the following diagram:
\[\begin{tikzcd}
	\la{\overline{x}^{ \overline{f}}} p && \la{\overline{x}}  \ract{p}{\nu^{id, \tysf}; \overline{f}} \\
	\\
	\la{\overline{x}^{\overline{gf}}} p'  && \la{x} \ract{p'}{\nu^{\overline{fg}};\overline{gf}} 
	\arrow["{\nu^{id, \overline{f}}; \nu^{id, \overline{f}}_x \multimap o}", from=1-1, to=1-3]
	\arrow["{\theta_2; id}"', from=1-1, to=3-1]
	\arrow["{\nu^{id, \overline{fg}}; \nu^{id, \overline{gf}}_x \multimap o}"', from=3-1, to=3-3]
\end{tikzcd}\]
\end{proof}
We can do so by applying Lemma \ref{lem:actred}: there exist a term $  t $, morphisms $\zeta_1,h_1, \zeta_2  $ and ground morphisms $ \mu_1, \mu_2$ s.t. \[  \ract{p}{id; \overline{f}} \to^{ \zeta_1, \mu_1 ; id} t \qquad \qquad  \ract{p'}{id;\overline{gf}} \to^{\zeta_2, \mu_2; id} t \] \[\la{x} \ract{p}{id; \overline{f}} \to^{\zeta_1; id}  \la{x^{\mu_1}} t \to_{ id; \mu_1 \multimap o  } \la{x} t   \] \[ \la{x} \ract{p'}{id;\overline{gf}} \to^{\zeta_2; id}  \la{x^{\mu_2}} t \to_{ id; {\mu_2}  \multimap o  } \la{x} t  \]
We can then conclude since, by Lemma \ref{lem:actred} again, $  \zeta_1 \nu^{id, \overline{f}}_s = \zeta_2 \nu^{ id, \overline{gf}}_{s'} \theta_2 $ and $ {\mu_1} \nu^{\overline{f}}_x = {\mu_2} \nu^{\overline{fg}}_x $.

Let $ s = (\la{\overline{x}} s) \overline{q}  $ and $ t_1 =  \subst{s}{x}{\overline{q}}  $. We remark that the step $ s \to t_2 $ must be induced by either a step of $s $ or a step of $ \overline{q} $. Then we conclude by applying Lemma \ref{lem:morpsubs}.

Let $ s = \la{\overline{x}^{\overline{f}}} p $ and $ t_1 = \la{ \overline{x}^{\overline{gf}}} p' $ with $ p \to^{\theta;g; o } p'  $. Let $ t_2 = \la{ \overline{x}^{\overline{hf}} } p'' $ with $  p \to^{\theta', h; o } p''    $. By the IH we have that there exists $ t' $ s.t. $   p' \to^{\zeta_1,g_1; o} t'  $ and $ p'' \to^{\zeta_2,g_2; o} t'$ with $  \zeta \theta = \zeta' \theta'$ and $ g_1 g = g_2 h $. We then conclude by contextuality of the reduction.

Let $  s = ( \la{\overline{x}^{\overline{f}}} p ) \overline{q}$ and $ t_1 =  (  \la{\overline{x}^{\overline{gf}}} p') \overline{q} $ with $  p \to^{\theta, g; o} p'    $. there are three possible cases. 

 If $  p \to^{\theta', \overline{g'}; o} p'' $ with $ t_1 =  (\la{\overline{x}^{\overline{fg'}}} p'') \overline{q} $ we reason by cases on the second step. If its again a step for $ \la{\overline{x}^{\overline{f}}} p $ we conclude either by applying the IH or by definition of base exponential step. Otherwise
 $     \overline{q} \to^{\theta'; \overline{h}}  \overline{q}'$. In that case, we conclude in this way:
\[\begin{tikzcd}
	{( \la{\overline{x}^{\overline{f}}} p ) \overline{q}} && {( \la{\overline{x}^{\overline{fh}}} p ) \overline{q'}} \\
	\\
	{(\la{\overline{x}^{\overline{gf}}} p') \overline{q}} && {(\la{\overline{x}^{\overline{gfh}}} p') \overline{q'}}
	\arrow["{id \otimes \theta'; id}", from=1-1, to=1-3]
	\arrow["{\theta \otimes id; id}"', from=1-1, to=3-1]
	\arrow["{\theta \otimes id; id}", from=1-3, to=3-3]
	\arrow["{id  \otimes \theta'; id}"', from=3-1, to=3-3]
\end{tikzcd}\] 

If $   \la{\overline{x}^{\overline{f}}} p  \to \la{\overline{x}}  \ract{p}{id, \overline{f}} $. Then $   t_1 =  (\la{\overline{x}}  \ract{p}{id, \overline{f}} )[\nu^{id,\overline{f}}_x] \overline{q} $. There are two possible cases for the second step.
If $ t_2 = (\la{\overline{x}^{\overline{g}}} t') \overline{q}  $ with $ p \to_{\zeta, \overline{g}; o} t'$, we conclude by applying Lemma \ref{lem:actred}. If $ t_2 = (\la{\overline{x}^{\overline{fg}}} p ) \overline{q} $ with $ \overline{q} \to_{\theta; \overline{g}} \overline{q} '$, we apply again Lemma \ref{lem:actred} and conclude.

 \subsection*{Embedding of simple types}
 
 Given $ \Gamma \vdash M : A $, we denote as $ \mathsf{cart}^{\overline{n}^M}_\Gamma $ the canonical ground morphism \[ \mathsf{cart}^{\overline{n}^M}_\Gamma : \intt{\Gamma} \to \intt{\Gamma}^{\overline{n}^M}\]
 which consists of the sequence $ \mathsf{cart}_{A_1}^{n^M_1},\dots, \mathsf{cart}_{A_n}^{n^M_n} $ with $  \Gamma = A_1,\dots, A_k$.

\begin{lemma}\label{lemma:subemsi}
Let $ \Gamma , x : A \vdash M : B $ and $ \Gamma \vdash N : A $. We have that 
\[  \subst{\qt{M}}{x}{ \seq{\qt{N}^{n^{M}_x}}  } = \qt{ \subst{N}{x}{M}} \]
 \end{lemma}
 \begin{proof}
 By induction on $ M $.
 
  If $  M = x $ then $ \qt{ M} = x $ with $   x : \seq{\intt{A}} \vdash x : \intt{A}  $. Hence $ n^{M}_x = 1 $ and we can conclude since $ \subst{M}{x}{N} = N $.
 
 If $  M = \la{y} M $ with $ \Gamma, x : A, y :  B \vdash M' : C$, we have that $  \qt{M} = \la{y}^{\mathsf{cart}_{\intt{B}}^{n}} \qt{M'}  $ with $ \intt{\Gamma}^{\overline{n}^M}, x : \intt{A}^{m^M}, y : \intt{B}^{n^M} \vdash \qt{M'} : \intt{C} $ ad $ \mathsf{cart}_{\intt{B}}^{n} : \intt{B} \to \intt{B}^{n^M} $. By definition $ \subst{M}{x}{N} = \la{y} \subst{M'}{x}{N} $ and by IH we have that $ \subst{\qt{M'}}{x}{\seq{\qt{M}}^{m^{M}}} = \qt{ \subst{M}{x}{N}} $. We can then conclude by applying the IH and observing that for any simple type $ A$ and permutation $ \sigma $, we have that $  \sigma \mathsf{cart}^{n}_{\intt{A}} = \mathsf{cart}^{n}_{\intt{A}} $.
 
 If $ M = PQ $ with $ \Gamma , x : A \vdash M : B \Rightarrow C $ and $ \Gamma \vdash N : B$, we have that $ \qt{M} = \qt{\Gamma}^{\overline{n}^M}, x : \seq{\intt{A}}^{ m^P + m^Q  } \otimes \qt{\Gamma}^{\overline{n}^{N}} \vdash \qt{M} \seq{\qt{N}} : \intt{C} $. We have that $\subst{M}{x}{N} = \subst{P}{x}{N} \subst{Q}{x}{N} $. By IH we get that $   \subst{\qt{P}}{x}{\seq{\qt{N}}^{m^{P}}} = \qt{\subst{P}{x}{N}} $ and that $ \subst{\qt{Q}}{x}{\seq{\qt{N}}^{m^{Q}}} = \qt{\subst{Q}{x}{N}} $. We conclude again then by applying the IH and observing that   for any simple type $ A$ and permutation $ \sigma $, we have that $  \sigma \mathsf{cart}^{n}_{\intt{A}} = \mathsf{cart}^{n}_{\intt{A}} $.
 \end{proof}

We recall that simple types enjoy subject reduction: if $\Gamma \vdash M : A $ and $ M \to_\beta N $ then $  \Gamma \vdash N : A $. We can then type reduction steps as follows: $ \Gamma \vdash (M \to_\beta N) : A $.

 \begin{theorem}
Let $ \Gamma \vdash (M \to_\beta N) : A  $ then there exists a  cartesian resource term $ t$, ground morphisms $ \nu, \nu'$  s.t. \[   ( \intt{\Gamma}^{n_M} \vdash \qt{M} : \intt{A}) \toexp_{\nu; id} ( \delta \vdash t : \intt{A}) \toli_{\nu'; id} (  \intt{\Gamma}^{n_N} \vdash \qt{N} : \intt{A} ) \]
with $   (\nu'\nu) \mathsf{cart}_\Gamma^{n_M} = \mathsf{cart}_\Gamma^{n_N} .  $
\end{theorem}
\begin{proof}
By induction on the reduction step $ M \to_{\beta} N $. If $ M = (\la{x^A} P)Q $ and $ N = \subst{P}{x}{Q}  $ then $  \qt{(\la{x^A} P)Q    }  = (\la{ x^{\mathsf{cart}_A^n} } \qt{P}) \seq{\qt{Q}}$. We have that 
\[    (\la{ x^{\mathsf{cart}_A^n} } \qt{P}) \seq{\qt{Q}} \toexp (\la{ x^{id} } \qt{P}) \seq{\qt{Q}}^{[\mathsf{cart}_A^m]}   \]
and $ \seq{\qt{Q}}^{[\mathsf{cart}_A^m]} = \overbrace{\seq{ \qt{Q}, \dots, \qt{Q} }}^{m \text{ times}} $ then we have that 
\[ (\la{ x^{id} } \qt{P}) \seq{\qt{Q}}^{[\mathsf{cart}_A^m]}  \toli \subst{\qt{P}}{x}{\seq{ \qt{Q}, \dots, \qt{Q} }}\]
we can then conclude applying Lemma \ref{lemma:subemsi}: $\subst{\qt{P}}{x}{\seq{ \qt{Q}, \dots, \qt{Q} }} = \qt{\subst{P}{x}{Q}} $.

let $ M = \la{x} P $ with $ P \to_\beta P' $. By IH we have that 
\[  {\qt{P}} \toexp_{\nu, \nu_x; id} {t} \toli_{nu', \nu'_x; id} {\qt{P'}}      \]
then by definition of contextual closure for abstraction, we have that
\[ {\la{  x^{\mathsf{cart}^{n^P}_A}  } \qt{P}} \toexp_{ \nu; id } {\la{ x^{ \nu_x \mathsf{cart}^{n^P}_A}  } t} \toli {\la{ x^{ (\nu'_x \nu_x ) \mathsf{cart}^{n^P}_A}} \qt{P'} }\]
we can then conclude by applying the IH, since we get  $ {(\nu'_x \nu_x ) \mathsf{cart}^{n^P}_A} = {\mathsf{cart}_A^{n^{P'}}}$.

Let $ M = PQ $ with $ P \to_\beta Q $. By definition, we have that 
\[       \qt{  \Gamma \vdash PQ : A } = \intt{\Gamma}^{\overline{n}^P} \otimes \intt{\Gamma   }^{\overline{n}^Q} \vdash \qt{P}\seq{\qt{Q}} : \intt{B}      \]
By IH, there exists $ t' $ and $ \nu'_1, \nu'_2 $ s.t.
\[       \qt{P} \toexp_{\nu'_1; id} t' \toli_{\nu'_2; id} \qt{ P'  }                    \]
  By definition of exponential and linear reductions we have
  \[    \qt{P}\seq{\qt{Q}} \toexp_{\nu'_1 \otimes id; id}  t' \seq{\qt{Q}} \toli_{\nu'_2 \otimes id; id}  \qt{P'}\seq{Q}        \]
  we then set $ t = t' \seq{Q} $.
  
  The other application case is completely symmetric. The coherence condition is satisfied as a direct application of the IH.

\end{proof}

\begin{lemma}
If all bags appearing in $ \ty  $ are singletons, then there exists a simple type $ A$
 s.t. $ \intt{A} = \ty .$\end{lemma}
 \begin{proof}
 By induction on the structure of $ \ty$. If $ \ty = o$ the result is immediate. 
 If $  \ty = \seq{\ty'} \multimap o $, by IH we have that there exists $ A'$ and $ B   $ s.t. $ \intt{A} = \ty' $ and $\intt{B} = b $. By definition of coarse approximation of simple types we set $  A= A' \Rightarrow B$ and then conclude that $ \intt{A' \Rightarrow B } = \seq{\intt{A'}} \multimap \intt{B}  $.
 \end{proof}

\begin{proposition}
$  \gamma \vdash s : \ty$ is qualitative iff there exists a simply typed term $  \Gamma \vdash M : A $ s.t. $  \qt{M} = s $.
\end{proposition}
\begin{proof}
By induction on the structure of $ s $. The variable case follows from the former lemma, the other cases are direct consequences of the IH.
\end{proof}

\subsection*{Embedding of Intersection Types}

Given $ \Gamma \vdash M : A $, we denote as $ \mathsf{cart}^{\overline{n}^M}_\Gamma $ the canonical ground morphism \[ \mathsf{cart}^{\overline{n}^M}_\Gamma : \intt{\Gamma} \to \intt{\Gamma}^{\overline{n}^M}\]

Given a list $ \seqdots{\ty}{1}{k} $ and integers $ n_1,\dots, n_k $, we set $ \seqdots{\ty}{1}{k}^{\seqdots{n}{1}{k}} = \seq{\ty_1^{n_1},\dots, t_k^{n_k} }  $. We recall that idempotent intersection types enjoys subject substitution: if $  \Gamma, x : (A_1 \cap \cdots \cap A_k) \vdash M : B$ and $ \Gamma \vdash N : (A_1 \cap \cdots \cap A_k) $ then $ \Gamma \vdash \subst{M}{x}{N} : B . $
 
\begin{lemma}\label{lemma:subemint}
Let $ \Gamma , x : \tilde{A} = (A_1 \cap \cdots \cap A_k) \vdash M : B $ and $ \Gamma \vdash N : A_i $ for $  i \in [k]$. We have that 
\[  \subst{\qt{ M}^{\Gamma, \tilde{A}}_B}{x}{ (\qt{  N }_{\tilde{A}}^\Gamma)^{\vec{n}_{\tilde{A}}^M}   } \cong_{\sigma_{\qt{ M}^{\Gamma, \tilde{A}}_B,(\qt{  N }_{\tilde{A}}^\Gamma)^{\vec{n}_{\tilde{A}}^M}   }} \qt{ \subst{N}{x}{M} }_B^\Gamma . \]
 \end{lemma}
 \begin{proof}
 By induction on the structure of $M $. If $ M = x $ then $  \subst{M}{x}{N} = N  $. We have that 
 \[       \begin{prooftree}
 \hypo{  A \in \tilde{A} }
 \infer1{  x_1: \tilde{A}_1 ,\dots, x : \tilde{A}, \dots, x_n : \tilde{A}_n \vdash x : A    }
 \end{prooftree}      \]
 and $ \Gamma \vdash N : B$ for all $ B \in \tilde{B} $. In particular then, $\Gamma \vdash N : A $. We have that 
 \[  \qt{       \begin{prooftree}
 \hypo{  A \in \tilde{A} }
 \infer1{  x_1: \tilde{A}_1 ,\dots, x : \tilde{A}, \dots, x_n : \tilde{A}_n \vdash x : A    }
 \end{prooftree}        } = \begin{prooftree} 
 \infer0{ x : \seq{ \intt{A} } \vdash x : \intt{A}  }
 \end{prooftree}  \]
 We can then conclude, since $   \qt{  \subst{M}{x}{N}    }^\Gamma_B  = \qt{ N}^\Gamma_A   $. 
 
 If $ M = y $ with $ y \neq x $, then the result is immediate by definition of coarse approximation.
 
 Let $  M = \la{y} P $, then $  \subst{M}{x}{N} = \la{y} \subst{P}{x}{N}  $. We have that
 \[    \begin{prooftree}  
 \hypo{ \Gamma, x : \tilde{A}, y : \tilde{B} \vdash P : B }
 \infer1{    \Gamma, x: \tilde{A} \vdash \la{y} P : \tilde{B} \Rightarrow B        }
 \end{prooftree}                               \]
 
  by IH we have that \[ \subst{\qt{ P  }_{B}^{\Gamma, \tilde{A}, \tilde{B}} }{x}{(\qt{ N}_{\tilde{A}^\Gamma})^{\vec{n}^P_{\tilde{A}}}}  \cong \qt{ \subst{P}{x}{N}  }^{\Gamma, \tilde{B}}_{B}   \]
  
 Hence , in particular, we have that $ (\intt{\Gamma}^{\overline{n}^M} \otimes \intt{\Gamma}^{\overline{n}^N})^{[\sigma]} = \intt{ \Gamma  }^{\overline{n}^{\subst{M}{x}{n}}}$ where $ \sigma$ is the appropriate permutation induced by linear substitution. We conclude then by applying the IH, since \[\la{ x^{ \sigma \mathsf{cart}_{\title{B}}^{\vec{n}}  } } \subst{\qt{ P}^{\Gamma, \tilde{A}, \tilde{B}}_B}{x}{(\qt{  N}^{\Gamma, \tilde{B}}_{\tilde{A}})^{\vec{n}^P_x}}  \cong \la{ x^{\mathsf{cart}_{\title{B}}^{\vec{n}}  } } \qt{ \subst{P}{x}{N} }^{\Gamma, \tilde{B}}_B . \]
 
 If $ M = PQ $ then $  \subst{M}{x}{N} = \subst{P}{x}{N}\subst{Q}{x}{N} $. By definition we have that 
 \[   \qt{M}^{\Gamma, \tilde{A}}_B = \qt{P}^{\Gamma, \tilde{A}}_{\tilde{B } \Rightarrow B} \qt{Q}^{\Gamma}_{\tilde{B}}     \]
 for some intersection type $ \tilde{B } $. By definition of linear substitution: 
 \[    \subst{\qt{M}^{\Gamma, \tilde{A}}_B}{x}{ \qt{N}^{\vec{n}^{M}_{\tilde{B}}}  } = \subst{\qt{P}^{\Gamma, \tilde{A}}_{\tilde{B } \Rightarrow B}}{x}{  \qt{N}^{\vec{n}^{M}_{\tilde{B}}}_1 }  \subst{\qt{Q}^{\Gamma}_{\tilde{B}}}{x}{\qt{N}^{\vec{n}^{M}_{\tilde{B}}}_2} . \]
  We observe that $ \vec{n}^M_{\tilde{B}} = \vec{n}_1 \oplus \vec{n_2} $ with 
\[     \qt{N}^{\vec{n}^{M}_{\tilde{B}}}_1 = \qt{N}^{\vec{n}_1} \qquad \qt{N}^{\vec{n}^{M}_{\tilde{B}}}_2 =   \qt{N}^{\vec{n}_2}     \]  
   We can then apply the IH and get: 
 \[  \subst{\qt{P}^{\Gamma, \tilde{A}}_{\tilde{B } \Rightarrow B}}{x}{  \qt{N}^{\vec{n}^{M}_{\tilde{B}}}_1 } \cong_{\sigma_1} \qt{ \subst{P}{x}{N} }_{\tilde{B} \Rightarrow B}^\Gamma \]
 \[ \subst{\qt{Q}^{\Gamma}_{\tilde{B}}}{x}{\qt{N}^{\vec{n}^{M}_{\tilde{B}}}_2} \cong_{\sigma_2} \qt{\subst{Q}{x}{N}}^\Gamma_{\tilde{B}}  .\]  
 Hence, we conclude by definition of resource term isomorphism, since
 \[        \subst{\qt{P}^{\Gamma, \tilde{A}}_{\tilde{B } \Rightarrow B}}{x}{  \qt{N}^{\vec{n}^{M}_{\tilde{B}}}_1 }\subst{\qt{Q}^{\Gamma}_{\tilde{B}}}{x}{\qt{N}^{\vec{n}^{M}_{\tilde{B}}}_2} \cong_{(\sigma_1 \otimes \sigma_2) \tau}  \qt{ \subst{P}{x}{N} }_{\tilde{B} \Rightarrow B}^\Gamma   \qt{\subst{Q}{x}{N}}^\Gamma_{\tilde{B}}      \]
 where $\tau $ is the appropriate permutation induced by linear substitution.
 
 The intersection type case follows the same structure of the application one.
 \end{proof}

 We recall that also idempotent intersection types enjoys subject reduction: if $\Gamma \vdash M : A $ and $  M \to_\beta N$ then $ \Gamma \vdash N : A $. We can then again type reduction steps as follows: $ \Gamma \vdash (M \to_\beta N) : A $.

 \begin{theorem}
Let $ \Gamma \vdash (M \to_\beta N) : A  $ then there exists a  cartesian resource terms $ t_1, t_2$, ground morphisms $ \nu_1, \nu_2$  s.t. \[   ( \intt{\Gamma}^{n_M} \vdash \qt{M} : \intt{A}) \toexp_{\nu; id} ( \delta_1 \vdash t_1 : \intt{A}) \toli_{\nu'; id} (  \delta_2 \vdash t_2 : \intt{A} ) \]
with $ t_2 \cong \qt{\Gamma \vdash N : A}  $ and $   (\sigma \nu'\nu) \mathsf{cart}_\Gamma^{n_M} = \mathsf{cart}_\Gamma^{n_N} .  $
\end{theorem}
\begin{proof}
By induction on the reduction step $ M \to_\beta N  $. Let $ M = (\la{x}P)Q $ and $ N = \subst{M}{x}{N} $. By definition we have that 
\[   \qt{M}^\Gamma_A = \intt{ \Gamma }^{\overline{n}^{ \la{x}P}} \otimes                 \intt{\Gamma}^{\overline{n}^Q} \vdash (\la{x^{\bigoplus \mathsf{cart}_{A_i}^{n_i^P}}} \qt{P  }^{\Gamma}_{(A_1 \cap \cdots \cap A_k)} ) \seq{\qt{Q}_{A_1}^\Gamma, \dots, \qt{Q}_{A_k}^\Gamma} \]
By definition of exponential reduction, we have that 
\[    {\qt{M}^\Gamma_A} \toexp_{ id \otimes ((\mathsf{cart}_{A_1}^{n_1}),\dots, (\mathsf{cart}_{A_k}^{n_k}) ); id}  {t_1 = (  (\la{x} \qt{P  }^{\Gamma}_{(A_1 \cap \cdots \cap A_k)} ) \seq{(\qt{Q}_{A_1}^\Gamma)^{n_1^Q}, \dots, (\qt{Q}_{A_k}^\Gamma)^{n_k^Q}}}  )\]
then we have that, by definition of linear reduction
\[    t_1 \toli_{\sigma; id} t_2 = ( \subst{\qt{P  }^{\Gamma}_{(A_1 \cap \cdots \cap A_k)}}{x}{\seq{(\qt{Q}_{A_1}^\Gamma)^{n_1^Q}, \dots, (\qt{Q}_{A_k}^\Gamma)^{n_k^Q}}} )      \]
By Lemma \ref{lemma:subemint} we can conclude that
\[  t_2 \cong_\sigma \qt{ N  }_A^\Gamma\]
with $  \sigma (id \otimes ((\mathsf{cart}_{A_1}^{n_1}),\dots, (\mathsf{cart}_{A_k}^{n_k}) )) = \mathsf{cart}_\Gamma^{\overline{n}^N} $.

Let $ M = \la{x} P $ and $ N = \la{x} P' $ with $  P \to P'$.

By IH we have that there exist $t'_1,t'_2, \nu'_1, \nu'_2 , \sigma'$ s.t. 
\[ \qt{P}_A^{\Gamma, \tilde{A}} \toexp_{\nu'_1; id} t'_1 \toli_{\nu'_2; id} t'_2 \]
with $ t'_2 \cong_{\sigma'} \qt{N}_A^{\Gamma, \tilde{A}} . $
Hence
\[ \qt{\la{x} P}_{\tilde{A} \Rightarrow A}^\Gamma = \la{x}^{\mathsf{cart}_{\tilde{A}}^{\vec{n}}}    \]

Then
The application case is a direct consequence of the IH.
\end{proof}

\begin{remark}\label{rem:unifprof}
We observe that if $ \gamma \vdash s : \ty = \qt{\Gamma \vdash M : A}  $ then $  s \lhd M$ and $ \gamma \lhd \Gamma, \ty \lhd A $. Moreover, if $ s \coh s' $ and $ s \lhd M , s' \lhd M'$ then $ M = M' $.
\end{remark}


\begin{proposition}
$\gamma \vdash s : \ty $ is a uniform term iff there exists an idempotent intersection typing  $ \Gamma \vdash M : A $ s.t. $ \qt{\Gamma \vdash M : A} = \gamma \vdash s : \ty   $.
\end{proposition}
\begin{proof}
$ (\Rightarrow) $ by induction on the structure of $s $, exploiting Remark \ref{rem:unifprof}. The list case is the only non trivial one. Let $ s = \seqdots{s}{1}{k} $ with 
\[     \begin{prooftree}
\hypo{ \gamma_1 \vdash s_1 : \ty_1 \ \cdots \ \gamma_k \vdash s_k : \ty_k }
\infer1{   \gamma_1 \otimes \cdots \otimes \gamma_k \vdash \seqdots{s}{1}{k} : \seqdots{\ty}{1}{k}     }
\end{prooftree}   \]
We have that $  s_i $ is uniform. Then, by IH, there exists $ M_i $ s.t. $ \qt{M_i} = s_i $. Since $ s $ is uniform, we have that $ s_i \coh s_j $ for all $ i,j$. Then we can conclude by Remark \ref{rem:unifprof}.

$(\Leftarrow)$ By induction on the structure of $M $. The interesting case is the introduction of the intersection type:
\[   \begin{prooftree}
\hypo{ \Gamma \vdash M : A_1 \ \cdots \ \Gamma \vdash M : A_k          }
\infer1{    \Gamma \vdash M : (A_1 \cap \cdots \cap A_k)   }
  \end{prooftree}        \]
  Then, by IH, $ \qt{\Gamma \vdash M : A_i} $ is uniform. We observe that if $ s,s' \lhd M $ then $ s \coh s' .$ We can then conclude that $ \seqdots{s}{1}{k}$ is uniform.
 \end{proof}



\end{document}